\documentclass[11pt,a4paper]{article}

\usepackage{amsfonts,amsmath,amssymb,amsthm}
\usepackage{latexsym,amscd}
\usepackage{amsbsy}
\usepackage{CJK}
\usepackage{indentfirst,mathrsfs}
\usepackage{latexsym,amscd}
\usepackage{amsbsy}
\usepackage[noadjust]{cite}
\usepackage{color}
\usepackage{multirow}
\usepackage{makecell}
\usepackage{float}
\usepackage{booktabs}
\usepackage{tikz}
\usetikzlibrary{positioning}
\usepackage{xcolor}
\usepackage{wrapfig}
\newtheorem{theorem}{Theorem}

\newtheorem{lem}{Lemma}

\newtheorem{prop}{Proposition}

\newtheorem{example}{Example}
\newtheorem{remark}{Remark}

\newcommand{\bF}{{\mathbb F}}

\newcommand{\C}{{\mathcal{C}}}

\begin{document}

\title{The Weight Hierarchies of Linear Codes from Simplicial Complexes}

\author{Chao Liu,\,\, Dabin Zheng,\,\, Wei Lu,\,\, Xiaoqiang Wang
\thanks{C. Liu, D. Zheng, W. Lu, X. Wang are with Hubei Key Laboratory of Applied Mathematics, Faculty of Mathematics and Statistics, Hubei University, Wuhan 430062, China.
 Email: chliuu@163.com; dzheng@hubu.edu.cn; d20230130@hubu.edu.cn; waxiqq@163.com. The corresponding author is Dabin Zheng. The research of Dabin Zheng and Xiaoqiang Wang was supported by NSFC under Grant
 Numbers 62272148, 11971156, 12001175.}}

\date{}
\maketitle

\begin{abstract}
The study of the generalized Hamming weight of linear codes is a significant research topic in coding theory as it conveys the structural information of the codes and determines their performance in various applications. However, determining the generalized Hamming weights of linear codes, especially the weight hierarchy, is generally challenging. In this paper, we investigate the generalized Hamming weights of a class of linear code $\C$ over $\bF_q$, which is constructed from defining sets. These defining sets are either special simplicial complexes or their complements in $\bF_q^m$. We determine the complete weight hierarchies of these codes by analyzing the maximum or minimum intersection of certain simplicial complexes and all $r$-dimensional subspaces of $\bF_q^m$, where $1\leq r\leq {\rm dim}_{\bF_q}(\C)$.
\end{abstract}

\vskip 6pt
\noindent {\it Keywords.} Linear code, generalized Hamming weight, simplicial complex, weight hierarchy.
\vskip6pt
\noindent {\it  2010 Mathematics Subject Classification.}\quad  94B05, 94B15

\vskip 30pt

\section{Introduction}

Let $q$ be a power of a prime number and ${\mathbb F}_q$ be a finite field with $q$ elements. An $[n,k,d]$ linear code $\C$ over the finite field  $\bF_q$  is a $k$-dimensional subspace of $\bF_q^n$
with length $n$ and minimum Hamming distance $d$. For each $r$ with $1\leq r\leq k$, let $[\C, r]$ be the set of all $r$-dimensional $\bF_q$-subspaces of $\C$. For each $H\in [\C, r]$, the {\em support} of $H$, denoted by supp($H$), is the set of not-always-zero bit positions of $H$, i.e.,
$$ {\rm supp}(H)=\left\{ i \, :\, 1\leq i\leq n, \,\, c_i\neq 0  \,\, {\rm for \,\, some} \,\, (c_1,c_2,\cdots,c_n)\in H \right\}.$$
The $r$-th {\em generalized Hamming weight} (GHW) of $\C$ is defined by
	$$ d_r(\C)= {\rm min} \left\{ \,|{\rm supp}(H)|\,:\, H \in [ \C, r]\, \right\},$$
where $|{\rm supp}(H)|$ denotes the cardinality of the set ${\rm supp}(H)$. The set $\left\{ d_1(\C), d_2(\C),\cdots, d_k(\C)\right\}$ is called the {\em weight hierarchy}
of $\C$. It is worth noting that $d_1(\C)$ is just the minimum Hamming distance of~$\C$.

The notion of GHWs was introduced by Helleseth et al.~\cite{Hellesethetal1977} and Kl{\o}ve~\cite{Klove1978}, which was a natural generalization of the minimum Hamming distance of linear codes. The GHWs of linear codes provide fundamental information of linear codes, which is important in many applications. In 1991, Wei~\cite{Wei1991} first gave a series of beautiful results on GHWs of linear codes and used them to characterize the cryptography performance of linear codes over the wire-tap channel of type II. From then on, the GHWs of linear codes have become an interesting research object in both theory and applications. The GHWs were used to deal with t-resilient functions, and the trellis or branch complexity of linear codes~\cite{Chor1985,Geometric}. Apart from these cryptographic applications, the GHWs also were used to computation of the state and branch complexity profiles of linear codes~\cite{Forney1994,Kasamietal1993}, indication of efficient ways to shorten linear codes~\cite{Hellesethetal19951}, determination of the erasure list-decodability of linear codes~\cite{Guruswami2003}, etc..

In the past two decades, an important progress has been made in the study of generalized Hamming weights of linear codes. Authors in~\cite{Ashikhmin1999,Cohen1994,Hellesethetal19951} gave the general lower and upper bounds on GHWs of linear codes. The values of GHWs have been determined or estimated for many classes of linear codes, such as Hamming codes~\cite{Wei1991}, Reed-Muller codes~\cite{Heijnen1998,Wei1991}, binary Kasami codes~\cite{Hellesethetal19952}, Melas and dual Melas codes~\cite{vandergeer1994}, some BCH codes and their duals \cite{Cheng1997,Duursma1996,Feng1992,Moreno1998,Shim1995,vandergeer19942,vandergeer1995}, some trace codes~\cite{Cherdien2001,Stichtenoth1994,vandergeer19952,vandergeer1996}, some algebraic geometry codes \cite{Yang1994} and cyclic codes \cite{Lishuxing2017,Ge2016,Yang2015}. However, it is very difficult to determine the complete weight hierarchies of linear codes, and to the best of our knowledge, only a few linear codes have known complete weight hierarchies so far.

On the other hand, the parameters and Hamming weight distributions of linear codes constructed using definition sets have been extensively and deeply studied in recent years, leading to fruitful results. Li and Mesnager~\cite{LiMesnager2020} surveyed the results on linear codes constructed from cryptographic functions. Recently, there has been some attention paid to the generalized Hamming weights of linear codes constructed from defining sets. Jian et al.~\cite{Jian2017} first investigated the generalized Hamming weights of linear codes constructed from skew sets. Li~\cite{Li2018}, Li~\cite{Li2021} and Liu et al.~\cite{LiuZheng2023}, Liu and Wang~\cite{LiuWang2020}, Li and Li~\cite{LiLi2021,LiLi2022} and Li et al.~\cite{Likangquan2022} investigated the generalized Hamming weights of certain classes of linear codes from defining sets. These sets are associated with quadratic functions, cyclotomic classes, or cryptographic functions. Compared to the abundant findings on the weight distributions of linear codes constructed by definition sets, there are very few results on the generalized Hamming weight hierarchies of this class of linear codes.

It is both interesting and challenging to find a suitable defining set of linear codes such that the Hamming weight hierarchy of the corresponding linear codes can be derived. Recently, some classes of linear codes constructed from simplicial complexes were deeply studied and their weight distributions were obtained~\cite{Chang2018,Hyun2020,Hu2023}. In this paper, we study the generalized Hamming weights of the linear codes constructed from the simplicial complexes of $\bF_q^m$, and determine the complete weight hierarchies of these classes of linear codes, whose defining sets are some special simplicial complexes or their complements in $\bF_q^m$. The key point is to determine the maximum or minimum intersection of the simplicial complexes and all $r$-dimensional subspaces of the discussed codes for $1\leq r\leq k$, where $k$ is the dimension of the object code.

The rest of this paper is organized as follows: In Section 2, we introduce simplicial complexes in $\bF_q^m$ and the linear codes from the defining-sets and their generalized Hamming weights. Section 3 determines the weight hierarchy of several a class of linear codes from simplicial complexes. These complexes are generated by the support sets of one, two, three, or a finite number of maximal elements of the simplicial complex, respectively. Section 4 determines the weight hierarchy of a class of linear codes from the complements of simplicial complexes, which are generated by the support sets of one, two, or a finite number of maximal elements, respectively. Finally, Section 5 concludes the paper.

\section{Preliminaries}\label{sec:preli}

\subsection{Linear codes constructed from defining-sets}

Let $q$ be a power of a prime and $\bF_q$ be a finite field of $q$ elements. Let ${\rm Tr}$ denote the trace function from ${\mathbb F}_{q^m}$ to ${\mathbb F}_{q}$.
For a set $D=\{d_1,d_2,\cdots, d_n\}\subset {\mathbb F}_{q^m}$, Ding and Niederreiter~\cite{Ding2007} proposed a generic method for constructing linear codes as follows:
\begin{equation*}\label{defingset}
\C_D=\left\{ {\bf c}(x) = \left( {\rm Tr}( xd_1), {\rm Tr}(xd_2), \cdots, {\rm Tr}(xd_n) \right), \,\, x\in {\mathbb F}_{q^m} \right\}.
\end{equation*}
Then, $\C_D$ is a linear code over $\bF_q$ of length $n$. The set $D$ is called the defining set of $\C_D$. This construction is fundamental since every liner code over $\bF_q$ can be represented as $\C_D$ for some defining set $D$~\cite{Xiang2016}. Recently, many good linear codes have been constructed by choosing proper defining sets $D$. It is known that $\bF_{q^m}$ is isomorphic to $\bF_q^m$. Fixing a basis of $\bF_{q^m}$ over $\bF_q$, any element $a\in \bF_{q^m}$ is identical to its coordinate vector ${\bf a}\in \bF_q^m$. So, the set $D \subset {\mathbb F_{q^m}}$ is identical to a set ${\bf D} \subset {\mathbb F_q^m}$ with respect to this basis. It is easy to show that $\C_D$ is equivalent to
\begin{equation}\label{defineset}
\C_{\bf D}=\{{\bf c}({\bf x})=({\bf x}\cdot {\bf d})_{{\bf d}\in {\bf D}} : {\bf x} \in {\mathbb F_q^m}\},
\end{equation}
where ${\bf D} = \{\bf d_1, d_2, \cdots, d_n\}$. Let $G$ be an $m\times n$ matrix with ${\bf d_i}'s$ as its columns, i.e., $G=({\bf d_1, d_2, \cdots, d_n})$. Then $G$ is the generator matrix of
$\C_D$ if the rank of $G$ is $m$. Please refer to \cite{Xiang2016} for a detailed proof.

\subsection{Simplicial complexes of $\bF_q^m$}

We introduce the notion of simplicial complexes of ${\mathbb F_q^m}$. For two vectors ${\bf u}=(u_1,u_2,\dots,u_m)$ and ${\bf v}=(v_1,v_2,\dots,v_m)$ in ${\mathbb F_q^m}$, we say that ${\bf u}$ covers ${\bf v}$, denoted ${\bf v}\preceq {\bf u}$, if ${\rm Supp}({\bf v})\subseteq {\rm Supp}({\bf u})$, where ${\rm Supp}({\bf u})= \{ i\, :\, u_i\neq 0,\, 1\leq i\leq m\}$. A subset $\Delta$ of ${\mathbb F_q^m}$ is called a simplicial complex if
${\bf u}\in \Delta$ and ${\bf v} \preceq {\bf u}$ imply that ${\bf v}\in \Delta$. An element ${\bf u}$ in $\Delta$ is said to be maximal if there is no element ${\bf v}\in {\Delta}$ such that ${\rm Supp}({\bf u}) \subset {\rm Supp}({\bf v})$. For a simplicial complex $\Delta$ of ${\mathbb F_q^m}$, the distinct maximal elements in $\Delta$ may have the same support set and this is different from the case in $\bF_2^m$. Let $ S_1, S_2,\dots, S_{\ell}$ be the support sets of all maximal elements of $\Delta$, and these support sets are uniquely determined by $\Delta$, denoted ${\rm supp}(\Delta)$. We say that $\Delta$ is generated by ${\rm supp}(\Delta)=\{ S_1, S_2,\dots, S_{\ell}\}$, that is,
\[ \Delta = \langle S_1, S_2,\dots, S_{\ell} \rangle = \left\{ {\bf u}\in \bF_q^m \,:\, {\rm supp} ({\bf u})\,\, {\rm is \,\, contained \,\, in \,\, some \,\,} S_i, 1\leq i\leq \ell \right\} .\]

If ${\rm supp}(\Delta)=\emptyset$, then set $\Delta=\{0\}$. If ${\rm supp}(\Delta)$ has only one element, such as $S_1$, then it is easy to verify that $\Delta=\langle S_1\rangle$ is a subspace of $\bF_q^m$ with dimension $|S_1|$. In general, if ${\rm supp}(\Delta)=\{ S_1, S_2,\dots, S_{\ell}\}$, then
$$\Delta = \langle S_1, S_2,\dots, S_{\ell} \rangle = \cup_{i=1}^\ell \langle S_i \rangle,$$
which is the union of $\ell$ subspaces of $\bF_q^m$ with dimensions $|S_i|$.

Note that the above definition of simplicial complexes of $\bF_q^m$ is from~\cite{Hu2023} and it is a generalization of the definition of simplicial complexes of $\bF_2^m$~\cite{Hyun2020}, and is different from the two definitions given in~\cite{Hyun2019} for $\bF_p^m$ and in~\cite{Shi2021} for $\bF_3^m$.

\begin{example}
Let $\Delta=\langle S_1, S_2 \rangle$ be a simplicial complex of ${\bF_q^4}$, where $S_1=\{1, 2\}$ and $S_2=\{2, 3\}$. Then
$$\langle S_1\rangle=\left\{ (a,b,0,0)\,:\, a, b\in \bF_q \right\}\,\, {\rm and } \,\, \langle S_2\rangle=\left\{ (0,c,d,0)\,:\, c, d\in \bF_q \right\}. $$
It is clear that $\langle S_1\rangle$  and $\langle S_2\rangle$ are  $2$-dimensional subspaces of ${\bF_q^4}$ and $\Delta = \langle S_1\rangle  \cup \langle S_2\rangle.$
\end{example}

\subsection{Generalized Hamming weights of linear codes from defining-set constructions}

In this subsection, we recall the generalized Hamming weights of linear codes obtained from the defining-set constructions. Let $\Delta=\langle S_1, S_2, \cdots, S_\ell \rangle$ be a simplicial complex of
$\bF_q^m$ and $\C_{\Delta}$ be a linear code defined in (\ref{defineset}) with dimension $k$. It is easy to see that $k=m- {\rm dim}_{\bF_q} K$, where
$$K=\left\{ {\bf x} \in {\mathbb F_q^m}\, :\, ( {\bf u} \cdot {\bf x} )_{{\bf u} \in \Delta} = 0 \right\}.$$
Li, Chen and Qu in~\cite{Likangquan2022} provided a formula for calculating the generalized Hamming weight of $\C_D$ as follows.
\begin{prop}[\cite{Likangquan2022}]{\label{GHW}}
Let $K$ be as mentioned earlier, and let $\C_{\Delta}$ be a linear code defined in (\ref{defineset}) with dimension $m-{\rm dim}_{\bF_q} K$. Then, for any integer $1\leq r\leq m-{\rm dim}_{\bF_q} K$, the $r$-th generalized Hamming weight of $\C_{\Delta}$ is
\begin{equation}\label{eq:rghw}
d_r(\C_\Delta) =n-{\rm max}\, \left\{ | \Delta \cap H^{\perp}| : H \in [{\mathbb F_q^m}, r]\,\, {\rm and} \,\, H\cap K=\{\bf 0 \} \right\},
\end{equation}
where $[{\mathbb F_q^m}, r]$ denotes the set of all $r$-dimensional subspaces of ${\mathbb F_q^m}$ and $H^{\perp}\subset \bF_q^m$ is the dual subspace of $H$.
\end{prop}

\begin{remark}
Let $\Delta$ be a simplicial complex and ${\rm Supp}(\Delta) = \{ S_1, S_2, \cdots, S_\ell \}$. Let $[m]$ denote the set $\{1, 2, \cdots, m\}$. If $ \cup_{i=1}^\ell S_i \subsetneq [m]$, then set $S^\prime = [m]\setminus \cup_{i=1}^\ell S_i$, and one can verify that $K = \langle S^\prime \rangle$, where $\langle S^\prime \rangle$ is a simplicial complex generated by the support set $S^\prime$, which is a subspace of $\bF_q^m$.
To determine $d_r(\C_\Delta)$, from (\ref{eq:rghw}) we need to calculate the maximum intersection of $\Delta$ and $H^\perp$ for all $r$-dimensional subspaces $H\subset \bF_q^m$ satisfying $H\cap K=\{ \bf 0\}$. For convenience, we assume that $\cup_{i=1}^\ell S_i=[m]$, i.e., the dimension of $\C_\Delta$ equals $m$, then $K=\{\bf 0 \}$ and $H\cap K=\{ \bf 0\}$ for any $r$-dimensional subspace of ${\mathbb F_q^m}$. In the following, we discuss the generalized Hamming weights of $\C_\Delta$ under this condition, which does not affect the conclusions.
\end{remark}

\section{The weight hierarchy of a linear code $\C_{\Delta}$}

In this section, we will discuss the weight hierarchy of a class of linear codes with the following form:
\begin{equation}{\label{definingset3}}
\C_{\Delta}=\left\{ {\bf c}({\bf x})=\left( {\bf x}\cdot {\bf v} \right)_{{\bf v}\in \Delta} : {\bf x} \in {\mathbb F_q^m} \right\},
\end{equation}
where $\Delta$ is a simplicial complex of ${\mathbb F_q^m}$ and ${\bf x}\cdot {\bf v}= \sum_{i=1}^m x_iv_i$ represents the standard inner product of vectors. By Proposition~\ref{GHW}, the key to obtain the $r$-th generalized Hamming weight of $\C_{\Delta}$ is to determine the maximum size of the intersections of the duals of all $r$-dimensional subspace $H$ with the simplicial complex $\Delta$. Let $\Delta=\langle S_1, S_2, \dots, S_\ell \rangle$, where $S_i$ is a support set of some maximal element of $\Delta$ for $i=1, 2, \cdots, \ell$ and $\cup_{i=1}^\ell S_i =[m]$. From Section~\ref{sec:preli} we know that $\Delta= \cup_{i=1}^\ell \langle S_i\rangle$, where $\langle S_i \rangle$ is a simplicial complex generated by the support set $S_i$, which is an $|S_i|$-dimensional subspace of $\bF_q^m$. Since $\cup_{i=1}^\ell S_i =[m]$ one can see that $K=\{ {\bf x} \in {\mathbb F_q^m}\, :\, ( {\bf u} \cdot {\bf x} )_{{\bf u} \in \Delta} =0 \}=\{\bf 0 \}$. Then, by Proposition~\ref{GHW} and the principle of inclusion-exclusion we have

\begin{equation}\label{eq:drc}
d_r(\C_{\Delta})=|\Delta| - \mathop{ \rm max}\limits_{H\in [\bF_q^m, r]} \left( \sum_{t=1}^\ell (-1)^{t-1} \sum\limits_{1\leq i_1<\cdots<i_t\leq \ell} | \cap_{j=1}^t \langle S_{i_j}\rangle \cap H^\perp|\right),
\end{equation}
where $H^{\perp}$ is the dual space of $H$. From this formula, it is difficult to determine the exact value of $d_r(\C_{\Delta})$ for a general simplicial complex $\Delta$. But the weight hierarchy of $\C_{\Delta}$ can be explicitly determined when $\Delta$ is in some special cases.

First, we consider the case that ${\Delta}$ is generated by a support set $S$ of one maximal element. Since $S=[m]$ we know that
$\Delta={\mathbb F_q^m}$. So, the following result is obvious.
\begin{theorem}{\label{thm1}}
Let $\Delta=\langle S \rangle$ be a simplicial complex of ${\mathbb F_q^m}$ with $S=[m]$, where $S$ is a support set of some maximal element of $\Delta$. Let $\C_{\Delta}$ be a linear code defined in (\ref{definingset3}). Then
$\C_{\Delta}$ is a $[q^m, m, q^m-q^{m-1}]$ code and has the weight hierarchy as follows:
\begin{equation*}
d_r(\C_{\Delta})= q^m -q^{m-r},\,\, 1\leq r\leq m.
\end{equation*}
\end{theorem}

\begin{example}
Let $\C_{\Delta}$ be a binary linear code  defined in (\ref{definingset3}), where $\Delta=\langle S \rangle $ is a simplicial complex of ${\mathbb F_2^4}$ generated by the support set $S=\{1,2,3,4\}$. By the help of Magma, we know that $\C_{\Delta}$ is an $[16, 4, 8]$ code and has the weight hierarchy as follows: $d_1=8, d_2=12, d_3=14,d_4=15$.
This result is consistent with Theorem \ref{thm1}.
\end{example}

Next, we consider the case that ${\Delta}$ is generated by the support sets of two maximal elements.

\begin{theorem}{\label{thm2}}
Let $\Delta=\langle S_1, S_2\rangle$ be a simplicial complex of ${\mathbb F_q^m}$ with $S_1\cup S_2=[m]$, where $S_1$ and $S_2$ are the support sets of two maximal elements of $\Delta$, and
$1\leq |S_1|\leq |S_2|<m$. Let $\C_{{\Delta}}$ be a linear code defined in~(\ref{definingset3}). Then, $\C_{\Delta}$ is a $[q^{|S_1|}+ q^{|S_2|}-q^{|S_1\cap S_2|}, m, q^{|S_1|}-q^{|S_1\cap S_2|+m-|S_2|-1}]$ code
and has the weight hierarchy given in Table~\ref{table:twomaxied}.	
\begin{table}[h]
    \centering
    \footnotesize
    \renewcommand{\arraystretch}{1.5}
    \setlength{\tabcolsep}{5pt}
	\caption{The weight hierarchy of $\C_{\Delta}$ in Theorem~\ref{thm2}}
    \label{table:twomaxied}
	\scalebox{1}{
   \begin{tabular}{ll}
	 \toprule
	\textbf{The weight hierarchy} & \textbf{Rang of $r$} \\
	\midrule
	$q^{|S_1|}-q^{|S_1\cap S_2|+m-r-|S_2|}$ & $1\leq r\leq m-|S_2|$ \\
	$q^{|S_1|}+q^{|S_2|}-q^{|S_1\cap S_2|}-q^{m-r}$ & $m-|S_2|\leq r\leq m$\\
	\bottomrule
\end{tabular}
}
\end{table}
\end{theorem}

\begin{proof}
It is known that $\Delta= \langle S_1\rangle \cup \langle S_2\rangle$, where $\langle S_i\rangle$ is a subspace of $\bF_q^m$ with dimension $|S_i|$ for $i=1, 2$.
It is clear that the dimension of $\C_\Delta$ is $m$ and its code length is equal to
$$ n=|\Delta|=q^{|S_1|}+q^{|S_2|}-q^{|S_1\cap S_2|}.$$
By applying the principle of inclusion-exclusion, we have for any $H\in [\bF_q^m, r]$ that
\begin{equation*}\label{twosimcom}
\begin{split}
|\Delta \cap H^{\perp}|&=|\langle S_1\rangle \cap H^{\perp}|+|\langle S_2\rangle\cap H^{\perp}|-|\langle S_1\rangle \cap \langle S_2\rangle\cap H^{\perp}|\\
                       &=|\langle S_1\rangle \cap H^{\perp}|+|\langle S_2\rangle\cap H^{\perp}|-|\langle S_1\cap S_2\rangle \cap H^{\perp}|,
\end{split}
\end{equation*}
where $H^{\perp}$ denotes the dual space of $H$ and $\langle S_1\cap S_2\rangle$ is a simplicial complex with one maximal element generated by the support set $S_1\cap S_2$, i.e., $\langle S_1\cap S_2 \rangle$ is a subspace of $\bF_q^m$ with dimension $|S_1\cap S_2|$. Next, we discuss the maximum value of $|\Delta \cap H^{\perp}|$ for all $r$-dimensional subspaces $H$ of $\bF_q^m$.

Let $|S_1\cap S_2|=t$ ($0\leq t< |S_1|$). We choose $i_1$ linearly independent vectors from $\langle S_1\rangle \setminus \langle S_1\cap S_2 \rangle $ and $i_2$ linearly independent vectors from
$\langle S_2\rangle \setminus \langle S_1\cap S_2 \rangle$ and $j$ linearly independent vectors from $\langle S_1\cap S_2\rangle$, and these $i_1+i_2+j$ vectors constitute a basis of $H^\perp$. Then,
\begin{equation*}\label{eq:twocap}
|\Delta\cap H^\perp| =|\langle S_1\rangle\cap H^{\perp}|+|\langle S_2\rangle\cap H^{\perp}|-|\langle S_1\cap S_2\rangle\cap H^{\perp}|=q^{i_1+j}+q^{i_2+j}-q^{j}=q^j(q^{i_1} + q^{i_2} -1).
\end{equation*}
From this equation, we know that the value of $|\Delta\cap H^\perp|$ is maximized when $i_1$ and $i_2$ take the largest possible values. Since the maximum values of $i_1+j$ and $i_2+j$ are $|S_1|$ and
$|S_2|$, respectively, and $|S_1|\leq |S_2|$, we have that whether $|\Delta \cap H^{\perp}|$ can reach its maximum value depends on whether $i_2+j$ reaches its maximum value.

If $|S_2| \leq m-r$, i.e., $1\leq r\leq m-|S_2|$, then we can choose an $H\in [\bF_q^m, r]$ such that $H^{\perp}=\langle S_2 \rangle + H^{\prime}$, where $H^{\prime} \subseteq \langle S_1\rangle \setminus \langle S_1\cap S_2\rangle $ is a $(m-r-|S_2|)$-dimensional subspace. Such a subspace exists since
$${\rm dim}_{\bF_q}\,H^{\prime}=m-r-|S_2| = |S_1|-t-r<|S_1|-t.$$
Then, this $H$ maximizes $|\Delta\cap H^\perp| $ for all $H\in [\bF_q^m, r]$, i.e.,
\[ \mathop{{\rm max}}\limits_{H\in [\bF_q^m, r]} \left( |\Delta\cap H^\perp| \right) = q^{m-r-|S_2| +|S_1\cap S_2|}+ q^{|S_2|} - q^{|S_1\cap S_2|} .\]

If $0\leq m-r\leq |S_2|$, i.e., $ m-|S_2| \leq  r\leq m$, then we can choose an $H\in [\bF_q^m, r]$ such that $H^\perp \subseteq \langle S_2\rangle$. In this case, $i_2+j$ can reach its maximum value of $m-r$. Then,
\[ \mathop{{\rm max}}\limits_{H\in [\bF_q^m, r]} \left(|\Delta\cap H^\perp| \right) = q^{m-r }.\]

From (\ref{eq:drc}) and these two cases, we obtain the Table~\ref{table:twomaxied}.

\end{proof}

\begin{example}
Let $\C_{\Delta}$ be a binary linear code defined in (\ref{definingset3}), where $\Delta=\langle S_1,S_2 \rangle $ is a simplicial complex of ${\mathbb F_2^5}$ generated
by the the support sets $S_1=\{1,2,3\}$ and $S_2=\{3,4,5\}$. By the help of Magma, we have that $\C_{\Delta}$ is a $[14, 5, 4]$ code and has the weight hierarchy as follows: $d_1=4, d_2=6, d_3=10,d_4=12, d_5=13$. This result is consistent with Table \ref{table:twomaxied} in Theorem \ref{thm2}.
\end{example}


Then, we discuss the weight hierarchy of $\C_{\Delta}$ for the case that ${\Delta}$ is generated by the support sets of three maximal elements.

\begin{theorem}{\label{thm3}}
Let $\Delta=\langle S_1,S_2,S_3\rangle$ be a simplicial complex of ${\mathbb F_q^m}$  with $S_1\cup S_2\cup S_3=[m]$, where $S_1, S_2$ and $S_3$ are the support sets of three maximal elements of $\Delta$, and
$1\leq |S_1|< |S_2|< |S_3|<m$. Then, the linear code $\C_{{\Delta}}$ defined in (\ref{definingset3}) has parameters $[q^{|S_1|}+q^{|S_2|}+q^{|S_3|}-q^{|S_1\cap S_2|}-q^{|S_1\cap S_3|}-q^{|S_2\cap S_3|}+q^{|\cap_{i=1}^3 S_i|}, m, q^{|S_1|}-q^{m-1-|S_2|-|S_3|+|S_1\cap S_2|+|S_1\cap S_3|+|S_2\cap S_3|-|\cap_{i=1}^3 S_i|}]$. Moreover,
\begin{itemize}
	\item [{\rm (1)}] if $|S_1\cap S_3|\leq |S_2\cap S_3|$, then $\C_{{\Delta}}$  has the weight hierarchy given in Table \ref{table:3maxied};
	\item [{\rm (2)}] if $|S_1\cap S_3|\geq |S_2\cap S_3|$, then $\C_{{\Delta}}$  has the weight hierarchy given in Table \ref{APP2}.
\end{itemize}
\begin{table}[h]
     \centering
    \footnotesize
    \renewcommand{\arraystretch}{1.2}
    \setlength{\tabcolsep}{3pt}
	\caption{The weight hierarchy of $\C_{\Delta}$ in Theorem~\ref{thm3}}
    \label{table:3maxied}
	\scalebox{1}{
	\begin{tabular}{ll}
	\toprule
	\textbf{The weight hierarchy} & \textbf{Rang of $r$}\\
	\midrule
    \vspace{0.2cm}
	$q^{|S_1|}-q^{m-r-|S_2|-|S_3|+|S_1\cap S_2|+|S_1\cap S_3|+|S_2\cap S_3|-|\cap_{i=1}^3 S_i|}$ & $1\leq r\leq m-|S_2|-|S_3|+|S_2\cap S_3|$ \\
   \vspace{0.2cm}
	$q^{|S_1|}+q^{|S_2|}-q^{m-r-|S_3|+|S_2\cap S_3|}-q^{|S_1\cap S_2|+|S_1\cap S_3|-|\cap_{i=1}^3 S_i|}$ &  \makecell[l]{$m-|S_2|-|S_3|+|S_2\cap S_3|\leq r\leq $
	\\$m-|S_3|-|S_1\cap S_2|+|\cap_{i=1}^3 S_i|$}\\
     \vspace{0.2cm}
	\makecell[l]{$q^{|S_1|}+q^{|S_2|}-q^{|S_1\cap S_2|}-q^{m-r-|S_3|+|S_1\cap S_3|}-q^{m-r-|S_3|+|S_2\cap S_3|}$ \\
	$+q^{m-r-|S_3|+|\cap_{i=1}^3 S_i|}$}  & \makecell[l]{ $m-|S_3|-|S_1\cap S_2|+|\cap_{i=1}^3 S_i| $\\
	 $\leq r\leq m-|S_3|$} \\
    \vspace{0.2cm}
	\makecell[l]{$q^{|S_1|}+q^{|S_2|}+q^{|S_3|}-q^{|S_1\cap S_2|}-q^{|S_1\cap S_3|}-q^{|S_2\cap S_3|}+q^{|\cap_{i=1}^3 S_i|}$\\
	 $-q^{m-r}$} & $m-|S_3|\leq r \leq m$\\
	\bottomrule
\end{tabular}
}
\end{table}

 \begin{table}[h]
     \centering
    \footnotesize
    \renewcommand{\arraystretch}{1.2}
    \setlength{\tabcolsep}{3pt}
	\caption{The weight hierarchy of $\C_{\Delta}$ inTheorem \ref{thm3}}
    \label{APP2}
	\scalebox{1}{
	\begin{tabular}{ll}
	\toprule
	\textbf{The weight hierarchy} & \textbf{Rang of $r$}\\
	\midrule
    \vspace{0.2cm}
	$q^{|S_1|}-q^{m-r-|S_2|-|S_3|+|S_1\cap S_2|+|S_1\cap S_3|+|S_2\cap S_3|-|\cap_{i=1}^3 S_i|}$ & $1\leq r\leq m-|S_2|-|S_3|+|S_2\cap S_3|$ \\
	\vspace{0.2cm}
	$q^{|S_1|}+q^{|S_2|}-q^{m-r-|S_3|+|S_2\cap S_3|}-q^{|S_1\cap S_2|+|S_1\cap S_3|-|\cap_{i=1}^3 S_i|}$ &  \makecell[l]{$m-|S_2|-|S_3|+|S_2\cap S_3|\leq r\leq $
	\\$m-|S_1|-|S_3|+|S_2\cap S_3|$}\\
	\vspace{0.2cm}
	\makecell[l]{$q^{|S_2|}-q^{m-r-|S_1|-|S_3|+|S_1\cap S_2|+|S_1\cap S_3|+|S_2\cap S_3|-|\cap_{i=1}^3 S_i|}$ } & \makecell[l]{ $m-|S_1|-|S_3|+|S_2\cap S_3| $\\
	 $\leq r\leq  m-|S_1|-|S_3|+|S_1\cap S_3|$} \\
    \vspace{0.2cm}
	$q^{|S_1|}+q^{|S_2|}-q^{|S_1\cap S_2|+|S_2\cap S_3|-|\cap_{i=1}^3 S_i|}-q^{m-r-|S_3|+|S_1\cap S_3|}$ & \makecell[l]{$m-|S_1|-|S_3|+|S_1\cap S_3|\leq r\leq $ \\
	$ m-|S_3|-|S_1\cap S_2|+|\cap_{i=1}^3 S_i|$ }   \\
	\vspace{0.2cm}
	\makecell[l]{$q^{|S_1|}+q^{|S_2|}-q^{|S_1\cap S_2|}-q^{m-r-|S_3|+|S_1\cap S_3|}-q^{m-r-|S_3|+|S_2\cap S_3|}$ \\
	$+q^{m-r-|S_3|+|\cap_{i=1}^3 S_i|}$ } & \makecell[l]{$m-|S_3|-|S_1\cap S_2|+|\cap_{i=1}^3 S_i| $ \\
	$\leq r\leq m-|S_3|$ } \\
	\vspace{0.2cm}
	\makecell[l]{$q^{|S_1|}+q^{|S_2|}+q^{|S_3|}-q^{|S_1\cap S_2|}-q^{|S_1\cap S_3|}-q^{|S_2\cap S_3|}+q^{|\cap_{i=1}^3 S_i|}$\\
	 $-q^{m-r}$} & $m-|S_3|\leq r \leq m$\\
	\bottomrule
\end{tabular}
}
\end{table}
\end{theorem}

\begin{proof}
It is known that $\Delta=\langle S_1 \rangle \cup \langle S_2 \rangle \cup \langle S_3 \rangle $, where $\langle S_i \rangle$ is a subspace of ${\mathbb F_q^m}$ with dimension $|S_i|$ for $i=1,2,3$.
Since $S_1\cup S_2\cup S_3=[m]$, the dimension of $\C_\Delta$ is $m$. By the principle of inclusion-exclusion, the code length of $\C_\Delta$ is equal to
$$n=|\Delta|=q^{|S_1|}+q^{|S_2|}+q^{|S_3|}-q^{|S_1\cap S_2|}-q^{|S_1\cap S_3|}-q^{|S_2\cap S_3|}+q^{|\cap_{i=1}^3 S_i|}.$$
From equation (\ref{eq:drc}), for any $H\in [{\mathbb F_q^m}, r]$, we have
\begin{equation}{\label{threesimcom}}
\begin{split}
|\Delta \cap H^{\perp}| &=\sum_{i=1}^3 |\langle S_i \rangle \cap H^{\perp}|-\sum_{1\leq i<j\leq 3}|\langle S_i \rangle \cap \langle S_j \rangle \cap H_r^{\perp}|+|\cap_{i=1}^3 \langle S_i \rangle \cap H_r^{\perp}| \\
&=\sum_{i=1}^3 |\langle S_i \rangle \cap H^{\perp}|-\sum_{1\leq i<j\leq 3}|\langle S_i \cap S_j \rangle \cap H_r^{\perp}|+|\langle \cap_{i=1}^3 S_i \rangle \cap H_r^{\perp}|,
\end{split}
\end{equation}
where $H^{\perp}$ denotes the dual space of $H$ and $\langle S_i \cap S_j \rangle, \langle \cap_{i=1}^3 S_i \rangle$ are simplicial complexes generated by the support sets $S_i \cap S_j$ and $\cap_{i=1}^3 S_i$, respectively. Next, we discuss the maximal value of $|\Delta \cap H^{\perp}|$ for all $r$-dimensional subspaces $H$ of ${\mathbb F_q^m}$.

We choose $i_1, i_2, i_3$ and $g$ linearly independent vectors from $\langle S_1\cap S_2\rangle \setminus \langle \cap_{i=1}^3 S_i\rangle, \langle S_1\cap S_3\rangle \setminus \langle \cap_{i=1}^3 S_i\rangle, \langle S_2\cap S_3\rangle \setminus \langle \cap_{i=1}^3 S_i\rangle$ and $\langle \cap_{i=1}^3 S_i \rangle$, respectively and choose $j_1,j_2$ and $j_3$ linearly independent vectors from $\langle S_1\rangle \setminus (\langle S_1\cap S_2\rangle +\langle S_1\cap S_3 \rangle), \langle S_2\rangle \setminus (\langle S_1\cap S_2\rangle +\langle S_2\cap S_3 \rangle)$ and  $\langle S_3\rangle \setminus (\langle S_1\cap S_3\rangle +\langle S_2\cap S_3 \rangle)$, respectively, and these $i_1+i_2+i_3+g+j_1+j_2+j_3$ vectors constitute a basis of a subspace $H^{\perp}$. From (\ref{threesimcom}) we have
\begin{equation}\label{eq:3simcom}
	\begin{split}
		|\Delta \cap H^{\perp}|&=q^{i_1+i_2+g+j_1}+q^{i_1+i_3+g+j_2}+q^{i_2+i_3+g+j_3}-q^{i_1+g}-q^{i_2+g}-q^{i_3+g}+q^{g}\\
		&=q^{g}(q^{i_1+i_2+j_1}+q^{i_1+i_3+j_2}+q^{i_2+i_3+j_3}-q^{i_1}-q^{i_2}-q^{i_3}+1).
	\end{split}
\end{equation}
The exponents $i_1+i_2+g+j_1, i_1+i_3+g+j_2, i_2+i_3+g+j_3$ are the dimensions of the intersections between $\langle S_1 \rangle $ and $H^\perp$,  $\langle S_2 \rangle $ and $H^\perp$, $\langle S_3 \rangle$ and $H^\perp$, respectively. From (\ref{eq:3simcom}) we know that $|\Delta \cap H^{\perp}|$ is maximized when the exponents $i_1+i_2+g+j_1$, $i_1+i_3+g+j_2$, $i_2+i_3+g+j_3$ reach their largest possible values. It is known that the maximum values of these exponents are $|S_1|$, $|S_2|$, $|S_3|$, respectively. Since dim$_{\bF_q} (H^\perp) < m$, these three exponents can not reach the maximum values at the same time. Since $|S_1|<|S_2|<|S_3|$, by (\ref{eq:3simcom}) we have that the maximum value of $|\Delta \cap H^{\perp}|$ depends on the exponents $i_1+i_3+g+j_2$ and $i_2+i_3+g+j_3$. In the following we discuss the maximum value of $|\Delta \cap H^{\perp}|$ according to the dimension of $H^\perp$.

If $|S_2|+|S_3|-|S_2\cap S_3| \leq m-r$, i.e., $1\leq r\leq m-|S_2|-|S_3|+|S_2\cap S_3|$. Then we can choose an $H \in [{\mathbb F_q^m}, r]$ such that $H^{\perp}=\langle S_2\rangle +\langle S_3 \rangle + H^{1}$,
where $H^{1} \subseteq \langle S_1 \rangle \setminus (\langle S_1\cap S_2 \rangle + \langle S_1\cap S_3 \rangle )$ has dimension
$${\rm dim}_{\bF_q}\,H^{1}=m-r-|S_2|-|S_3|+|S_2\cap S_3|<|S_1|-|S_1\cap S_2|-|S_1\cap S_3|+|\cap_{i=1}^3 S_i|.$$
 Then, this $H$ maximizes $|\Delta \cap H^{\perp}|$ and
\[
\begin{split}
	\mathop{{\rm max}}\limits_{H\in [\bF_q^m, r]} \left(  |\Delta\cap H^\perp| \right) &= q^{|S_2|}+q^{|S_3|}+q^{|S_1\cap S_2|+|S_1\cap S_3|-|\cap_{i=1}^3 S_i|+m-r-|S_2|-|S_3|+|S_2\cap S_3|} \\
	&-q^{|S_1\cap S_2|}-q^{|S_1\cap S_3|}-q^{|S_2\cap S_3|}+q^{|\cap_{i=1}^3 S_i|}.
\end{split}
\]

If $|S_3|+ |S_1\cap S_2|- |S_1 \cap S_2 \cap S_3| \leq m-r\leq |S_2|+|S_3|-|S_2\cap S_3|$, i.e., $m-|S_2|-|S_3|+|S_2\cap S_3|\leq r\leq m-|S_3|-|S_1\cap S_2|+ |\cap_{i=1}^3 S_i|$. Then, to maximize $|\Delta \cap H^\perp|$, from (\ref{eq:3simcom}) the space $H^\perp$ should firstly contain $\langle S_3\rangle$ and $\langle S_1\cap S_2\setminus
\cap_{i=1}^3 S_i \rangle$, and the rest of $H^\perp$ is chosen from either $\langle S_2 \rangle $ or $\langle S_1 \rangle$ depending on the size of $|S_1\cap S_3|$ and
$|S_2\cap S_3|$.

{\bf (1)} The case $|S_1\cap S_3|\leq |S_2\cap S_3|$. In this case, the known intersection of $H^\perp$ and $\langle S_2 \rangle $ is not less than the known intersection of $H^\perp$ and $\langle S_1 \rangle $, i.e., $i_3+g\geq i_2+g$. So, the rest of $H^\perp$ should be chosen from $\langle S_2 \rangle \setminus (\langle S_1 \cap S_2 \rangle + \langle S_3 \cap S_2 \rangle)$. Hence, we can find an $H\in [\bF_q^m, r]$ such that
\[ H^\perp = \langle S_3 \rangle + \langle S_1\cap S_2 \setminus \cap_{i=1}^3 S_i \rangle + H^2,\]
where $H^2 \subseteq \langle S_2 \rangle \setminus (\langle S_1 \cap S_2 \rangle + \langle S_3 \cap S_2 \rangle)$ and has dimension
\[{\rm dim}_{\bF_q} H^2 = m-r -|S_3|- |S_1\cap S_2| + |\cap_{i=1}^3 S_i| \leq |S_2|-|S_1\cap S_2|-|S_2\cap S_3|+|\cap_{i=1}^3 S_i|.\]
This $H$ maximizes $|\Delta \cap H^{\perp}|$ and
\[ \begin{split}
	\mathop{{\rm max}}\limits_{H\in [\bF_q^m, r]} \left(  |\Delta\cap H^\perp| \right) &=q^{|S_3|}+q^{|S_1\cap S_2|+|S_1\cap S_3|-|\cap_{i=1}^3 S_i|}+q^{m-r-|S_3|+|S_2\cap S_3|} \\
	&-q^{|S_1\cap S_2|}-q^{|S_1\cap S_3|}-q^{|S_2\cap S_3|}+q^{|\cap_{i=1}^3 S_i|}.
\end{split}
\]

If $|S_3|\leq m-r\leq |S_3|+|S_1\cap S_2|-|\cap_{i=1}^3 S_i|$, i.e., $m-|S_3|-|S_1\cap S_2|+|\cap_{i=1}^3 S_i|\leq r\leq m-|S_3|$. In this case, from (\ref{eq:3simcom}) we know that
$H^\perp = \langle S_3\rangle + H^3$, where $H^3 \subseteq \langle S_1\cap S_2 \rangle \setminus \langle S_1\cap S_2 \cap S_3 \rangle$ can maximize $|\Delta\cap H^\perp|$ and
\[
\begin{split}
	\mathop{{\rm max}}\limits_{H\in [\bF_q^m, r]} \left(  |\Delta\cap H^\perp| \right) &=q^{|S_3|}+q^{|S_1\cap S_3|+m-r-|S_3|}+q^{|S_2\cap S_3|+m-r-|S_3|}\\
	&-q^{|\cap_{i=1}^3 S_i|+m-r-|S_3|}-q^{|S_1\cap S_3|}-q^{|S_2\cap S_3|}+q^{|\cap_{i=1}^3 S_i|}.
\end{split}
\]

If $0\leq m-r\leq |S_3|$, i.e., $m-|S_3|\leq r\leq m$. In this case, to maximize $|\Delta\cap H^\perp|$, by (\ref{eq:3simcom}) we know that $H^\perp$ should be chosen from $\langle S_3 \rangle$, then $\langle S_3\rangle \cap H^\perp = H^\perp$. Furthermore, $\langle S_i \rangle \cap H^\perp = \langle S_i \rangle \cap \langle S_3 \rangle \cap H^\perp$ for $i=1, 2$. From
(\ref{threesimcom}) we have
\[
\begin{split}
\mathop{{\rm max}}\limits_{H\in [\bF_q^m, r]} \left(  |\Delta\cap H^\perp| \right) &=q^{m-r}+q^{i_2+g}+q^{i_3+g}-q^{i_2+g}-q^{i_3+g}-q^{g}+q^{g}=q^{m-r}.
\end{split}
\]

From equation (\ref{eq:drc}) and the discussion above, we obtain  Table \ref{table:3maxied}.

{\bf (2)} The case $|S_1\cap S_3|\geq |S_2\cap S_3|$. It is known that ${\rm dim}_{\bF_q}(H^\perp)=m-r$. When $|S_3|+ |S_1\cap S_2|- |S_1 \cap S_2 \cap S_3| \leq m-r \leq |S_2|+|S_3|-|S_2\cap S_3|$, to maximize $|\Delta\cap H^\perp|$,  the desired $H^\perp$ should first contain $\langle S_3 \rangle$ and $\langle S_1\cap S_2\rangle \setminus \langle S_1\cap S_2 \cap S_3\rangle$. In this case, it is easy to show that $|S_1|+|S_3|-|S_1\cap S_3|\leq |S_2|+|S_3|-|S_2\cap S_3|$. The selection of the remaining part of $H^\perp$ requires further analysis of dimension of $H^\perp$.

If $ |S_1|+ |S_3| - |S_2\cap S_3| \leq m-r \leq |S_2| + |S_3|-|S_2\cap S_3|$, i.e., $m-|S_2|-|S_3|+|S_2\cap S_3|\leq r\leq m-|S_1|-|S_3|+|S_2\cap S_3|$, then $m-r-|S_3|+|S_2\cap S_3|\geq |S_1|$. In this case, taking all of the rest of $H^\perp$ from $\langle S_2 \rangle $ makes the intersection of $H^\perp$ and $\langle S_2 \rangle$ greater than $q^{|S_1|}$. That is, we can find an $H \in [{\mathbb F_q^m},r]$ such that
$$H^{\perp}=\langle S_3\rangle +\langle S_1\cap S_2 \setminus S_1\cap S_2\cap S_3\rangle+H^{4},$$
 where $H^{4}\subseteq \langle S_2\rangle \setminus (\langle S_1\cap S_2\rangle +\langle S_2\cap S_3\rangle )$ and has dimension
$${\rm dim}_{\mathbb F_q}\,H^{4}=m-r-|S_3|-|S_1\cap S_2|+|S_1\cap S_2\cap S_3|<|S_2|-|S_1\cap S_2|-|S_2\cap S_3|+|S_1\cap S_2\cap S_3|.$$
This $H$ maximizes $|\Delta \cap H^{\perp}|$ in this case, and 	
\[\begin{split}
	\mathop{{\rm max}}\limits_{H\in [\bF_q^m, r]} \left(  |\Delta\cap H^\perp| \right) &=q^{|S_3|}+q^{|S_1\cap S_2|+|S_1\cap S_3|-|S_1\cap S_2\cap S_3|}+q^{m-r-|S_3|+|S_2\cap S_3|} \\
	&-q^{|S_1\cap S_2|}-q^{|S_1\cap S_3|}-q^{|S_2\cap S_3|}+q^{|S_1\cap S_2\cap S_3|}.
\end{split}
\]

If  $|S_1|+|S_3| -|S_1\cap S_3| \leq m-r \leq |S_1|+ |S_3| - |S_2\cap S_3|$, i.e., $m-|S_1|-|S_3|+|S_2\cap S_3|\leq r\leq m-|S_1|-|S_3|+|S_1\cap S_3|$, then
$m-r-|S_3|+|S_2\cap S_3|< |S_1|$. In this case, the value of the exponent of $i_1+i_3+g+j_2$ is less than $|S_1|$ if the vectors of $H^\perp \setminus \langle S_3\rangle$
are chosen from $\langle S_2\rangle$. To maximize $|\Delta\cap H^\perp|$, $H^{\perp}$ should firstly contain $\langle S_1\rangle +\langle S_3\rangle $ and the rest of $H^{\perp}$ is chosen from $\langle S_2\rangle $. That is,  we can find an $H \in [{\mathbb F_q^m},r]$ such that
$$H^{\perp}=\langle S_1\rangle +\langle S_3\rangle+H^{5},$$
where $H^{5}\subseteq \langle S_2\rangle \setminus (\langle S_1\cap S_2\rangle +\langle S_2\cap S_3\rangle )$ and has dimension
 $${\rm dim}_{\mathbb F_q}\, H^{5}=m-r-|S_1|-|S_3|+|S_1\cap S_3|<|S_2|-|S_1\cap S_2|-|S_2\cap S_3|+|S_1\cap S_2\cap S_3|.$$
 This $H$ maximizes $|\Delta \cap H^{\perp}|$ in this case, and 	
\[\begin{split}
	\mathop{{\rm max}}\limits_{H\in [\bF_q^m, r]} \left(  |\Delta\cap H^\perp| \right) &=q^{|S_3|}+q^{|S_1|}+q^{m-r-|S_1|-|S_3|+|S_1\cap S_2|+|S_2\cap S_3|-|S_1\cap S_2\cap S_3|} \\
	&-q^{|S_1\cap S_2|}-q^{|S_1\cap S_3|}-q^{|S_2\cap S_3|}+q^{|S_1\cap S_2\cap S_3|}.
\end{split}
\]

If $|S_3|+|S_1\cap S_2|-|S_1\cap S_2\cap S_3|\leq m-r\leq |S_1|+|S_3|-|S_1\cap S_3|$, i.e., $m-|S_1|-|S_3|+|S_1\cap S_3|\leq r\leq m-|S_3|-|S_1\cap S_2|+|S_1\cap S_2\cap S_3|$. To maximize $\Delta\cap H^\perp$, $H^{\perp}$ should firstly contain $\langle S_3 \rangle$ and $\langle S_1\cap S_2 \rangle \setminus \langle S_1\cap S_2\cap S_3\rangle$.
The rest of $H^\perp$ should be chosen from $\langle S_1 \rangle$ since $|S_1\cap S_3| \geq |S_2 \cap S_3|$. That is, we can find an $H$ such that
$$H^{\perp}=\langle S_3\rangle +  \langle S_1\cap S_2 \setminus  S_1\cap S_2\cap S_3 \rangle + H^6, $$
where $H^6 \subseteq \langle S_1\rangle \setminus (\langle S_1\cap S_2\rangle +\langle S_1\cap S_3\rangle )$ and has dimension
$${\rm dim}_{\mathbb F_q}\, H^6 =m-r-|S_3|-|S_1\cap S_2|+|S_1\cap S_2\cap S_3|<|S_2|-|S_1\cap S_2|-|S_2\cap S_3|+|S_1\cap S_2\cap S_3|.$$
This $H$ maximizes $|\Delta \cap H^{\perp}|$ in this case, and
\[\begin{split}
	\mathop{{\rm max}}\limits_{H\in [\bF_q^m, r]} \left(  |\Delta\cap H^\perp| \right) &=q^{|S_3|}+q^{|S_1\cap S_2|+|S_2\cap S_3|-|S_1\cap S_2\cap S_3|}+q^{m-r-|S_3|+|S_1\cap S_3|} \\
	&-q^{|S_1\cap S_2|}-q^{|S_1\cap S_3|}-q^{|S_2\cap S_3|}+q^{|S_1\cap S_2\cap S_3|}.
\end{split}
\]

If $ |S_3|\leq m-r\leq |S_3|+|S_1\cap S_2|-|S_1\cap S_2\cap S_3|$, i.e., $m-|S_3|-|S_1\cap S_2|+|S_1\cap S_2\cap S_3|\leq r\leq m-|S_3|$. In this case, we can find an $H$ such that
 $$H^{\perp}=\langle S_3\rangle +H^7,$$
where $H^7 \subseteq \langle S_1\cap S_2 \rangle \setminus \langle S_1\cap S_2 \cap S_3 \rangle$ and has dimension
 $${\rm dim}_{\mathbb F_q} H^7 = m-r-|S_3|<|S_1\cap S_2|-|S_1\cap S_2\cap S_3|.$$
This $H$ maximizes $|\Delta \cap H^{\perp}|$ in this case, and
\[
\begin{split}
	\mathop{{\rm max}}\limits_{H\in [\bF_q^m, r]} \left(  |\Delta\cap H^\perp| \right) &=q^{|S_3|}+q^{|S_1\cap S_3|+m-r-|S_3|}+q^{|S_2\cap S_3|+m-r-|S_3|}\\
	&-q^{|S_1\cap S_2\cap S_3|+m-r-|S_3|}-q^{|S_1\cap S_3|}-q^{|S_2\cap S_3|}+q^{|S_1\cap S_2\cap S_3|}.
\end{split}
\]

If $0\leq m-r\leq |S_3|$, i.e., $m-|S_3|\leq r\leq m$. In this case, to maximize $|\Delta\cap H^\perp|$, by (\ref{eq:3simcom}) the $H^\perp$ should be chosen from $\langle S_3 \rangle$, then $\langle S_3\rangle \cap H^\perp = H^\perp$. Furthermore, $\langle S_i \rangle \cap H^\perp = \langle S_i \rangle \cap \langle S_3 \rangle \cap H^\perp$ for $i=1, 2$. From
(\ref{threesimcom}) we have
\[
\begin{split}
\mathop{{\rm max}}\limits_{H\in [\bF_q^m, r]} \left(  |\Delta\cap H^\perp| \right) &=q^{m-r}+q^{i_2+g}+q^{i_3+g}-q^{i_2+g}-q^{i_3+g}-q^{g}+q^{g}=q^{m-r}.
\end{split}
\]

From equation (\ref{eq:drc}) and the result above, we obtain Table \ref{APP2}.

\end{proof}

\begin{example}
Let $\C_{\Delta}$ be a binary linear code  defined in (\ref{definingset3}), where $\Delta=\langle S_1,S_2,S_3 \rangle $ is a simplicial complex of ${\mathbb F_2^6}$ generated
by the support sets $S_1=\{1,2\}$, $S_2=\{1,3,4\}$ and $S_3=\{2,3,5,6\}$. By the help of Magma, we have that $\C_{\Delta}$ is a $[23,6,4]$ code and has the weight hierarchy  as follows: $d_1=4,\, d_2=7,\, d_3=15,\, d_4=19,\,  d_5=21,\, d_6=22$. This result is consistent with Table \ref{table:3maxied} in Theorem \ref{thm3}.
\end{example}

\begin{example}
Let $\C_{\Delta}$ be a linear code defined in (\ref{definingset3}), where $\Delta=\langle S_1,S_2,S_3 \rangle $ is a simplicial complex of ${\mathbb F_3^5}$ generated by the support sets $S_1=\{1,2\}$, $S_2=\{1,3,4\}$ and $S_3=\{2,3,4,5\}$. By the help of Magma, we know that $\C_{\Delta}$ is a $[103,5,22]$ code and has the weight hierarchy  as follows: $d_1=22,\, d_2=76,\, d_3=94,\, d_4=100,\, d_5=102$. This result is consistent with Table \ref{APP2} in Theorem \ref{thm3}.
\end{example}

\begin{remark}
Theorem~\ref{thm3} gives the weight hierarchy of a linear code $\C_{\Delta}$, where $\Delta=\langle S_1, S_2, S_3 \rangle $ is generalized by three support sets of some maximal elements of a simplicial complex of ${\mathbb F_q^m}$ under the condition that $|S_1|<|S_2|<|S_3|$. If any of them are equal, the results will be different, and these results are presented in Appendix.
\end{remark}

Theorems 1, 2 and 3 have determined the weight hierarchy of a linear code $\C_{\Delta}$ for a simplicial complex $\Delta$ of $\bF_q^m$ generalized by the support sets of one, two or three maximal elements of $\Delta$, respectively. It is difficult to determine the weight hierarchy of $\C_{\Delta}$ for a general simplicial complex $\Delta$. However, if the generating support sets of a simple complex $\Delta$ do not intersect with each other, then the weight hierarchy of the linear code $\C_{\Delta}$ can be derived as the following theorem.

\begin{theorem}{\label{thm4}}
Let $\Delta=\langle S_1,S_2,\dots,S_\ell \rangle$ be a simplicial complex of ${\mathbb F_q^m}$ with $\cup_{i=1}^\ell S_i=[m]$, where $1\leq |S_1|\leq |S_2|\leq\dots\leq |S_\ell|<m$ and $S_i\cap S_j=\emptyset$ for any $1\leq i<j\leq \ell $. Then, the linear code $\C_{\Delta}$  defined in (\ref{definingset3}) is a $[ \sum_{i=1}^\ell q^{|S_i|}-\ell+1, m, q^{|S_1|}-q^{m-1-\sum_{j=2}^\ell |S_j|}]$ code and has the weight hierarchy in Table~\ref{table:nmaxied}.
\begin{table}[h]
     \centering
    \footnotesize
    \renewcommand{\arraystretch}{1.2}
    \setlength{\tabcolsep}{3pt}
	\caption{The weight hierarchy of $\C_{\Delta}$ in Theorem~\ref{thm4}}
    \label{table:nmaxied}
	\scalebox{1}{
	\begin{tabular}{ll}
	\toprule
	\textbf{The weight hierarchy}\, &\, \textbf{Rang of $r$}\\
	\midrule
	$q^{|S_1|}-q^{m-r-|\cup_{i=2}^\ell S_{i}|}$ \,&\, $1\leq r\leq m-|\cup_{i=2}^\ell S_{i}|$ \\
	$q^{|S_1|}+q^{|S_2|}-q^{m-r-|\cup_{i=3}^\ell S_{i}|}-1$\,&\, $m-|\cup_{i=2}^\ell S_{i}|<r\leq m-|\cup_{i=3}^\ell S_{i}|$\\
	$\cdots$ \, &\, $\cdots$\\
	$\sum_{i=1}^j q^{|S_i|}-q^{m-r-|\cup_{i=j+1}^{\ell}S_i|}-j+1$\, &\, $m-|\cup_{i=j}^{\ell}S_i|<r\leq m-|\cup_{i=j+1}^{\ell}S_i|$\\
	$\cdots$ \,&\, $\cdots$\\
	$\sum_{i=1}^\ell q^{|S_i|}-q^{m-r}-\ell+1$ \,&\, $m-|S_\ell|<r\leq m$\\
	\bottomrule
\end{tabular}
}
\end{table}
\end{theorem}

\begin{proof}
It is known that $\Delta=\langle S_1\rangle \cup \langle S_2\rangle  \cup \cdots\cup \langle S_\ell \rangle $, where $\langle S_i \rangle $ is a subspace of ${\mathbb F_q^m}$ with dimension $|S_i|$ for $1\leq i\leq \ell$. Since $\cup_{i=1}^\ell S_i=[m]$, the dimension of $\C_{\Delta}$ is $m$. It is known that $S_i\cap S_j=\emptyset$ for any $1\leq i<j\leq \ell $. By the principle of inclusion-exclusion we have the length of the code $\C_{\Delta}$
$$n=|\Delta|=\sum_{i=1}^\ell q^{|S_i|}-\sum_{1\leq i<j\leq \ell }q^{|S_i\cap S_j|}+\cdots +(-1)^{\ell-1}q^{|S_1\cap S_2\cap \cdots \cap S_\ell|}=\sum_{i=1}^\ell q^{|S_i|}-\ell +1, $$
and for any $H\in [{\mathbb F_q^m}, r]$,
\begin{equation}{\label{tsim}}
\begin{split}
|\Delta \cap H^{\perp}| &=\sum_{i=1}^\ell |\langle S_i \rangle \cap H^{\perp}|-\sum_{1\leq i<j\leq \ell}|\langle S_i \rangle \cap \langle S_j \rangle \cap H^{\perp}|+\cdots +(-1)^{\ell-1}|\cap_{i=1}^\ell \langle S_i \rangle \cap H^{\perp}| \\
		&=\sum_{i=1}^\ell |\langle S_i \rangle \cap H^{\perp}|-\ell +1.
	\end{split}
\end{equation}
We choose $\delta_j$ linearly independent vectors from $\langle S_j \rangle$ for $1\leq j\leq \ell$, and these $\sum_{j=1}^\ell \delta_j$ vectors constitute a basis of  $H^{\perp}$.
From (\ref{tsim}) we have
$$|\Delta \cap H^{\perp}|= \sum_{i=1}^\ell q^{\delta_i}-\ell +1.$$
Hence, $|\Delta \cap H^{\perp}|$ reaches the maximum value when all exponents $\delta_j's$ are the maximum possible values.
It is known that the maximum possible value of $\delta_j$ is $|S_j|$, but these $t$ exponents can not reach its maximum value at the same time since ${\rm dim}_{\bF_q}(H^\perp) <m$. Then, we have the following cases.
		
If $|\cup_{i=2}^\ell {S_{i}}|\leq m-r\leq m-1$, i.e., $1\leq r\leq m-|\cup_{i=2}^\ell {S_{i}}|$, then we can choose an $H\in [{\mathbb F_q^m}, r]$ such that
$H^{\perp}= \sum_{i=2}^\ell \langle S_i \rangle +H^{\prime}$, where $H^{\prime} \subseteq \langle S_1 \rangle$, and
$${\rm dim}_{\bF_q} \,H^{\prime}=m-r-\sum_{i=2}^\ell |S_i|< |S_1|.$$
This $H$ maximizes $|\Delta \cap H^{\perp}|$ in this case, and
\[\begin{split}
	\mathop{{\rm max}}\limits_{H\in [\bF_q^m, r]} \left(  |\Delta\cap H^\perp| \right) &=q^{m-r-|\cup_{i=2}^\ell S_{i}|}+ \sum_{i=2}^\ell q^{|S_i|}-\ell +1.
\end{split}
\]

If $|\cup_{i=3}^\ell {S_{i}}|\leq m-r<|\cup_{i=2}^\ell {S_{i}}|$, i.e., $m-|\cup_{i=2}^\ell {S_i}|<r\leq m-|\cup_{i=3}^\ell {S_{i}}|$. We choose an $H\in [{\mathbb F_q^m}, r]$ such that $H^{\perp}= \sum_{i=3}^\ell \langle S_i \rangle + H^{*}$,
where $H^{*}\subseteq \langle S_2\rangle $, and 
$${\rm dim}_{\bF_q} \,H^{*}=m-r-\sum_{i=3}^\ell |S_i|< |S_2|.$$
This $H$ maximizes $|\Delta \cap H^{\perp}|$ in this case, and
\[\begin{split}
	\mathop{{\rm max}}\limits_{H\in [\bF_q^m, r]} \left(  |\Delta\cap H^\perp| \right) &=1+q^{m-r-|\cup_{i=3}^\ell {S_{i}}|}+ \sum_{i=3}^\ell q^{|S_i|}-\ell +1.
\end{split}
\]
Repeating the process above we obtain the maximum value of $ |\Delta\cap H^\perp|$ for the different dimension of $H$. From the equation (\ref{eq:drc}) and the known maximum value of $ |\Delta\cap H^\perp|$,
we obtain Table~\ref{table:nmaxied}.

%
	
\end{proof}


\begin{example}
Let $\C_{\Delta}$ be a linear code defined in (\ref{definingset3}), where $\Delta=\langle S_1,S_2,S_3,S_4 \rangle $ is a simplicial complex of ${\mathbb F_3^6}$ generated by the support sets $S_1=\{1\}$, $S_2=\{2\}$, $S_3=\{3,4\}$ and $S_4=\{5,6\}$. By the help of Magma, we know that $\C_{\Delta}$ is a $[21,6,2]$ code and has the weight hierarchy as follows: $d_1=2, d_2=4, d_3=10,d_4=12, d_5=18, d_6=20$. This result is consistent
with Table \ref{table:nmaxied} in Theorem~\ref{thm4}.
\end{example}

\section{The weight hierarchy of a linear code $\C_{\Delta^c}$}

In Section 3, we have determined the weight hierarchy of a linear code $\C_{\Delta}$ over $\bF_q$, where $\Delta$ is a simplicial complex of $\bF_q^m$. The simplicial complex $\Delta$ is generated by the support sets of some maximal elements of $\Delta$. In the following we discuss the weight hierarchy of a linear code over $\bF_q$ whose defining set is the complement of a simplicial complex of $\bF_q^m$. First, we give the definition of the linear codes that will be discussed as follows:
\begin{equation}{\label{definingset4}}
\C_{\Delta^c}=\left\{ {\bf c}({\bf x})=\left( {\bf x}\cdot {\bf v} \right)_{{\bf v}\in \Delta^c} : {\bf x} \in {\mathbb F_q^m} \right\}.
\end{equation}
where $\Delta^c={\bF_q^m}\setminus \Delta $ and $\Delta$ is a simplicial complex of $\bF_q^m$. It is difficult to determine the weight hierarchy of $\C_{\Delta^c}$ for a general simplicial complex $\Delta$ of $\bF_q^m$, and we will discuss this topic for some special $\Delta$ in this section.

Let $\Delta=\langle S_1, S_2, \cdots, S_{\ell}\rangle$ be a simplicial complex of ${\mathbb F_q^m}$, where $S_i$ for $1\leq i\leq \ell$ is the supports set of some maximal elements of $\Delta$. Let $\Delta^c = \bF_q^m \setminus \Delta$ be the complement of $\Delta$ in $\bF_q^m$. It has been shown in~\cite{Hu2023} that $\C_{\Delta^c}$ has dimension $m$, and then $K=\{ {\bf x} \in {\mathbb F_q^m}\, :\, ( {\bf u} \cdot {\bf x} )_{{\bf u} \in \Delta} =
 0\}= \{ {\bf 0} \}$. By Proposition~\ref{GHW} and the principle of inclusion-exclusion, the $r$-th generalized Hamming weight of $\C_{\Delta^c}$ is

\begin{equation}\label{eq:drc1}
d_r(\C_{\Delta^c})= |\Delta^c| - \mathop{\rm max}\limits_{H\in [\bF_q^m, r]} \left( q^{m-r}- \left( \sum_{t=1}^\ell (-1)^{t-1} \sum\limits_{1\leq i_1 < \cdots < i_t\leq \ell} | \cap_{j=1}^t \langle S_{i_j}\rangle \cap H^\perp|\right) \right).
\end{equation}

By (\ref{eq:drc1}), to obtain the $r$-th generalized Hamming weight of $\C_{\Delta^c}$ we need to determine the minimum value of $|\Delta \cap H^{\perp}|$ for all $r$-dimensional subspaces $H$ of $\bF_q^m$. This is different from determining the generalized Hamming weights of $\C_{\Delta}$ in Section~3. To this end, we first discuss the largest subspace contained in the remainder of the sum of two subspaces by removing the union of these two subspaces.

\begin{lem}{\label{V1V2}}
Let $U$ and $V$ be two subspaces of $\bF_q^m$ with dimensions $u$ and $v$, respectively. Then the largest subspace contained in $\{ {\bf 0} \} \cup (U+V)\setminus (U\cup V)$ has the dimension ${\rm min}\{ u-d, v-d\}$, where $d={\rm dim}_{\bF_q} (U\cap V)$.
\end{lem}

\begin{proof}
We can find an $U^\prime \subseteq U$ and an $V^\prime \subseteq V$ such that $U = U^\prime + U\cap V$ and $V=V^\prime + U\cap V$, and
$U^\prime \cap V^\prime =\{ {\bf 0} \}$. Then ${\rm dim}_{\bF_q} U^\prime = u-d$ and  ${\rm dim}_{\bF_q} V^\prime = v-d$. Without the loss of generality,
we assume that $u\leq v$. Let $ \alpha_1,\alpha_2,\dots,\alpha_{u-d}$ be a basis of $U^\prime$ and $\beta_1,\beta_2,\dots,\beta_{v-d}$
be a basis of $V^\prime$. It is easy to show that the vectors $\alpha_1+\beta_1, \alpha_2+\beta_2, \cdots,\alpha_{u-d}+\beta_{u-d}$ are linearly independent over $\bF_q$.
Set $W={\rm span}_{\bF_q}\{\alpha_1+\beta_1,\alpha_2+\beta_2, \cdots,\alpha_{u-d}+\beta_{u-d}\}$. This is a subspace of $U+ V$ with dimension $u-d$. One can show that $W\cap U=\{ {\bf  0} \}$ and $W\cap V=\{ {\bf 0} \}$. Then, we have found an $(u-d)-$dimensional subspace $W$ in $U+V$, but not in $U\cup V$.

For any subspace $W^*$ of $U+V$ with dimension greater than $u-d$, we can show that $W^*\cap V \neq \{ {\bf 0 } \}.$ It is clear that $W^* + V \subseteq U +V$. Then
$${\rm dim}_{\bF_q}(W^*) +{\rm dim}_{\bF_q}(V)- {\rm dim}_{\bF_q}(W^* \cap V)= {\rm dim}_{\bF_q} (W^*+ V) \leq {\rm dim}_{\bF_q} ( U + V) = u+v-d.$$
Since ${\rm dim}_{\bF_q}(W^*) \geq u-d+1$, we derive that ${\rm dim}_{\bF_q}(W^*\cap V) >0$. This completes the proof.
\end{proof}

\begin{remark}{\label{rem:V1Vn}}
The result in Lemma~\ref{V1V2} can be generalized to the case of a finite number of subspaces of $\bF_q^m$.
\end{remark}


First, we discuss the weight hierarchy of the linear code $\C_{\Delta^c}$, where $\Delta$ is generated by the support set of one maximal element of $\Delta$.

\begin{theorem}{\label{com1}}
Let $\Delta=\langle S \rangle$ be a simplicial complex of ${\mathbb F_q^m}$, where $S$ is a support set of some maximal element of $\Delta$ and $|S|<m$.
Then, the linear code $\C_{\Delta^c}$ defined in (\ref{definingset4}) is a $[q^m - q^{|S|}, m, q^m-q^{|S|}-q^{m-1}+q^{|S|-1} ]$ code and has the weight hierarchy in Table~\ref{com1}.
\begin{table}[h]
     \centering
    \footnotesize
    \renewcommand{\arraystretch}{1.2}
    \setlength{\tabcolsep}{3pt}
	\caption{The weight hierarchy of $\C_{\Delta}$ in Theorem~\ref{com1}}
    \label{table:com1maxied}
	\scalebox{1}{
	\begin{tabular}{ll}
	\toprule
	\textbf{The weight hierarchy} & \textbf{Rang of $r$}\\
	\midrule
	$q^m-q^{|S|}-q^{m-r}+q^{|S|-r}$ & $1\leq r\leq |S|$ \\

	$q^m-q^{|S|}-q^{m-r}+1$ & $|S|\leq r\leq m$\\
	\bottomrule
\end{tabular}
}
\end{table}
\end{theorem}

\begin{proof}
It is clear that the code length of $\C_{\Delta^c}$ is $|\Delta^c|=q^m-q^{|S|}$ and its dimension is $m$, which has been shown in~\cite{Hu2023}. Next, we show the $r$-th generalized Hamming weight of $\C_{\Delta}$. From equation (\ref{eq:drc1}), we have
\[ \begin{split}
d_r(\C_{\Delta^c})&=q^m-q^{|S|}-\mathop{\rm max}\limits_{H\in [\bF_q^m, r]} (q^{m-r}- |\langle S\rangle \cap H^{\perp}|)\\
                  &=q^m-q^{|S|}- q^{m-r}+\mathop{\rm min}\limits_{H\in [\bF_q^m, r]} ( |\langle S\rangle \cap H^{\perp}|).
\end{split}
\]
	
If $0\leq m-r\leq m-|S|$, i.e., $|S|\leq r\leq m$, then we can find a subspace $H\in [{\mathbb F_q^m}, r]$ such that $H^{\perp}$ is contained in the dual of $\langle S\rangle $.
In this case, we have $|\langle S\rangle \cap H^{\perp}|=1$, and $\mathop{\rm min}\limits_{H\in [\bF_q^m, r]} ( |\langle S\rangle \cap H^{\perp}|)=1$. So,
$d_r(\C_{\Delta^c}) = q^m-q^{|S|}- q^{m-r} +1$.

If $m-|S|\leq m-r$, i.e., $1\leq r\leq |S|$, then  $H^{\perp}$ and $\langle S\rangle $ must have a nontrivial intersection. In this case, we choose a subspace $H\in [{\mathbb F_q^m}, r]$ such that $H^{\perp}=\langle S\rangle^{\perp}+H^{\prime}$, where $\langle S\rangle^{\perp}$ is the dual of $\langle S\rangle$ and $H^{\prime}\subseteq \langle S \rangle$ is a $(|S|-r)$-dimensional subspace. Then,
$\mathop{\rm min}\limits_{H\in [\bF_q^m, r]} ( |\langle S\rangle \cap H^{\perp}|) =q^{|S|-r}.$
So, $d_r(\C_{\Delta^c}) = q^m-q^{|S|}- q^{m-r} +q^{|S|-r}$.
	
From equation (\ref{eq:drc1}), we obtain Table \ref{table:com1maxied}.
\end{proof}

\begin{example}
Let $\C_{\Delta^c}$ be a binary linear code defined in (\ref{definingset4}), where $\Delta=\langle S \rangle $ is a simplicial complex of ${\mathbb F_2^5}$ generated by the support set $S=\{2,3,4\}$. By the help of Magma, we know that $\C_{\Delta^c}$ is a $[24, 5, 12]$ code and has the weight hierarchy as follows: $d_1=12, d_2=18, d_3=21,d_4=23, d_5=24$. This result is consistent with Table \ref{table:com1maxied} in Theorem \ref{com1}.
\end{example}

Then, we discuss the weight hierarchy of the linear code $\C_{\Delta^c}$, where $\Delta$ is generated by two support sets of some maximal elements of $\Delta$.

\begin{theorem}{\label{com2}}
Let $\Delta=\langle S_1, S_2 \rangle$ be a simplicial complex of ${\mathbb F_q^m}$, where $S_1$ and $S_2$ are the support sets of two maximal element of $\Delta$, and $1\leq |S_1|\leq |S_2|<m$. Then, $\C_{\Delta^c}$ is a $[q^m-(q^{|S_1|}+q^{|S_2|}-q^{|S_1\cap S_2|}), m, q^m-q^{|S_1|}-q^{|S_2|}-q^{m-1}+q^{|S_1|-1}+q^{|S_2|-1} ]$ code and has the weight hierarchy in Table~\ref{com2}.
\begin{table}[h]
     \centering
    \footnotesize
    \renewcommand{\arraystretch}{1.2}
    \setlength{\tabcolsep}{3pt}
	\caption{The weight hierarchy of $\C_{\Delta}$ in Theorem~\ref{com2}}
    \label{table:com2maxied}
	\scalebox{1}{
	\begin{tabular}{ll}
	\toprule
	\textbf{The weight hierarchy} & \textbf{Rang of $r$}\\
	\midrule
	$q^m-q^{|S_1|}-q^{|S_2|}-q^{m-r}+q^{|S_1|-r}+q^{|S_2|-r}$ & $1\leq r<|S_1|-|S_1\cap S_2|$ \\
	$q^m-q^{|S_1|}-q^{|S_2|}+q^{|S_1\cap S_2|}-q^{m-r}+q^{|S_2|-r}$ & $|S_1|-|S_1\cap S_2|\leq r<|S_2|$\\
	$q^m-q^{|S_1|}-q^{|S_2|}+q^{|S_1\cap S_2|}-q^{m-r}+1$ & $|S_2|\leq r\leq m$ \\
	\bottomrule
\end{tabular}
}
\end{table}
\end{theorem}

\begin{proof}
It is easy to see that $\Delta= \langle S_1\rangle \cup \langle S_2\rangle$, where $\langle S_i\rangle$ is a subspace of $\bF_q^m$ with dimension $|S_i|$ for $i=1, 2$. The dimension of $\C_{\Delta^c}$ is $m$, which has been shown in \cite{Hu2023}, and the  length of code $\C_{\Delta^c}$ is
$$n=|\Delta^c|=q^m-|\Delta|=q^m-q^{|S_1|}-q^{|S_2|}+q^{|S_1\cap S_2|}.$$
By (\ref{eq:drc1}), the $r$-th generalized Hamming weight of $\C_{\Delta^c}$ is
\begin{equation}{\label{twosimcom0}}
\begin{split}
d_r(\C_{\Delta^c})&=|\Delta^c|-\mathop{\rm max}\limits_{H\in [\bF_q^m, r]} (q^{m-r}- |\Delta \cap H^{\perp}|)\\
&=q^m-q^{|S_1|}-q^{|S_2|}+q^{|S_1\cap S_2|}- q^{m-r}+\mathop{\rm min}\limits_{H\in [\bF_q^m, r]}\left( |\Delta \cap H^{\perp}|\right),
	\end{split}
	\end{equation}
where $H^{\perp}$ denote the dual space of $H$, and
\begin{equation}\label{eq:dhinters}
|\Delta \cap H^{\perp}|=|\langle S_1 \rangle \cap H^{\perp}|+|\langle S_2 \rangle \cap H^{\perp}|-|\langle S_1 \cap S_2 \rangle \cap H^{\perp}|.
\end{equation}
To obtain the $r$-th generalized Hamming weight of $\C_{\Delta^c}$, we need to determine the minimum value of $|\Delta \cap H^{\perp}|$ for all $r$-dimensional subspaces $H$.

It is known that $\Delta=\langle S_1\rangle \cup \langle S_2\rangle$, which is contained in $\langle S_1\rangle +\langle S_2\rangle$. Since $|S_1|\leq |S_2|$, by Lemma~\ref{V1V2} we know that $\{ {\bf 0 } \}\cup (\langle S_1\rangle +\langle S_2\rangle)\setminus (\langle S_1\rangle \cup \langle S_2\rangle)$ contains the largest subspace
$W$, which has the dimension $|S_1|-|S_1\cap S_2|$, and the dual $S^\perp$ of $\langle S_1\rangle +\langle S_2\rangle$ has the dimension $m-|S_1|-|S_2|+|S_1\cap S_2|$.
One can see that $W$ and $S^\perp$ have a trivial intersection and $W+S^\perp$ has the dimension $|S_1|-|S_1\cap S_2|+m-|S_1|-|S_2|+ |S_1\cap S_2|= m-|S_2|$.

If $m-r\leq m-|S_2|$, i.e., $ r \geq |S_2|$, then we can find an $H\in [\bF_q^m, r]$ such that
$H^\perp \subseteq W + S^\perp$. In this case, $\langle S_1 \rangle \cap H^\perp= \{ {\bf 0}\}$, $\langle S_2\rangle \cap H^\perp= \{ {\bf 0 } \}$ and
$\langle S_1\cap S_2\rangle \cap H^\perp= \{ {\bf 0 }\}$. From (\ref{eq:dhinters}) we konw that this $H$ minimizes $|\Delta \cap H^{\perp}|$, and
$\mathop{\rm min}\limits_{H\in [\bF_q^m, r]}\left( |\Delta \cap H^{\perp}|\right)= 1$. From (\ref{twosimcom0}) we have
\[ d_r(\C_{\Delta^c})=q^m-q^{|S_1|}-q^{|S_2|}+q^{|S_1\cap S_2|}- q^{m-r}+1. \]

If $m-|S_2|<m-r\leq m-|S_1|$, i.e., $|S_1| \leq r<|S_2|$, then for any $H\in [\bF_q^m, r]$,  $H^{\perp}\cap \langle S_2 \rangle \neq \{ {\bf 0} \}$ since
${\rm dim}_{\bF_q}\,H^{\perp}+{\rm dim}_{\mathbb F_q}\,\langle S_2\rangle =m-r+|S_2|>m$. It is known that
$${\rm dim}_{\mathbb F_q}\,(W+S^{\perp})=m-|S_2|, \ (W+S^{\perp})\cap \langle S_2\rangle=\{ {\bf 0 }\} \ {\rm and} \ (W+S^{\perp})\cap \langle S_1\rangle=\{{\bf 0} \}. $$
Then, we have
$$\langle S_2\rangle^{\perp}=W+S^{\perp}\ {\rm and}\ \langle S_2\rangle^{\perp}\subseteq \langle S_1\rangle^{\perp}.$$
So, we can find an $H\in [\bF_q^m, r]$ such that $H^\perp \subseteq \langle S_1 \rangle^\perp$ and $H^\perp = \langle S_2 \rangle^\perp + H^1$,
where $H^1 \subseteq \langle S_2 \rangle$ has the dimension $m-r-(m-|S_2|)=|S_2|-r$. This $H$ minimizes the intersections of $H^\perp$ and $\langle S_1\rangle$, and of $H^\perp$ and $\langle S_2\rangle$ at the same time, and
$$ |H^{\perp} \cap \langle S_1\rangle| =1,\,\, |H^{\perp}\cap \langle S_2\rangle|=  q^{|S_2|-r}, \,\, |H^{\perp} \cap \langle S_1 \cap S_2 \rangle |=1.$$
In this case, from (\ref{twosimcom0}) and (\ref{eq:dhinters}) we have
\[ d_r(\C_{\Delta^c})=q^m-q^{|S_1|}-q^{|S_2|}+q^{|S_1\cap S_2|}- q^{m-r}+ q^{|S_2|-r}. \]


If $m-|S_1|< m-r\leq m-|S_1|+|S_1 \cap S_2|$, i.e., $|S_1|-|S_1\cap S_2|\leq r<|S_1|$, then for any $H\in [{\mathbb F_q^m},r]$, $H^{\perp} \cap \langle S_1 \rangle \neq \{ \bf 0 \}$ and $H^{\perp} \cap \langle S_2 \rangle \neq \{ \bf 0 \}$ since
$${\rm dim}_{\bF_q}\,H^{\perp}+{\rm dim}_{\bF_q}\,\langle S_1\rangle =m-r+|S_1|>m\ {\rm and}\ {\rm dim}_{\bF_q}\,H^{\perp}+{\rm dim}_{\bF_q}\,\langle S_2\rangle =m-r+|S_2|>m,$$
respectively. It is known that $\langle S_2\rangle^{\perp}=W+S^{\perp}\ {\rm and}\ \langle S_2\rangle^{\perp}\subseteq \langle S_1\rangle^{\perp}$, then we can find an $H\in [\bF_q^m, r]$ such that
$$H^{\perp}=\langle S_1\rangle^{\perp}+H^{2}=\langle S_2\rangle^{\perp}+H^{3}+H^{2},$$
where $H^{2} \subseteq \langle S_1\rangle$, $\langle S_1\rangle^{\perp}=\langle S_2\rangle^{\perp}+H^{3}$ and $H^{3}\subseteq \langle S_2\rangle \setminus \langle S_1\cap S_2\rangle$, and
$${\rm dim}_{\bF_q}\,H^{2}=m-r-(m-|S_1|)=|S_1|-r,\,\, {\rm dim}_{\bF_q}\,H^{3}=m-|S_1|-(m-|S_2|)=|S_2|-|S_1|.$$
This $H$ minimizes the intersections of $H^{\perp}$ and $\langle S_1\rangle$, and of $H^{\perp}$ and $\langle S_2\rangle$ at the same time. We have
$$|\langle S_1\rangle \cap H^{\perp}|=q^{|S_1|-r},\,\, |\langle S_2\rangle \cap H^{\perp}|=q^{|S_2|-r}.$$
However, to minimize the value of $|\Delta \cap H^{\perp}|$, by (\ref{eq:dhinters}) the chosen $H$ need to further maximize the value of $|\langle S_1 \cap S_2 \rangle \cap H^{\perp}|$.
Then, we make  $H^{2}\subseteq \langle S_1\cap S_2\rangle \subseteq \langle S_1\rangle$. Hence, we have
 $$|\langle S_1\cap S_2\rangle \cap H^{\perp}|=q^{|S_1|-r}.$$
In this case, by (\ref{twosimcom0}) and (\ref{eq:dhinters}) we have
$$d_r(\C_{\Delta^c})=q^m-q^{|S_1|}-q^{|S_2|}+q^{|S_1\cap S_2|}-q^{m-r}+q^{|S_2|-r}.$$

If $m-|S_1|+|S_1\cap S_2|<m-r$, i.e., $1\leq r<|S_1|-|S_1\cap S_2|$, then we still choose an $H\in [{\mathbb F_q^m},r]$ such that
 $$H^{\perp}=\langle S_2\rangle^{\perp}+H^{3}+H^{4},$$
where $\langle S_1\cap S_2\rangle \subseteq H^{4}\subseteq \langle S_1\rangle$ and $H^{3}\subseteq \langle S_2\rangle \setminus \langle S_1\cap S_2\rangle$. It is easy to show that this $H$ minimizes the intersections
as follows:
$$|\langle S_1\rangle \cap H^{\perp}|=q^{|S_1|-r},\,\, |\langle S_2\rangle \cap H^{\perp}|=q^{|S_2|-r}\ {\rm and}\ |\langle S_1\cap S_2\rangle \cap H^{\perp}|=q^{|S_1\cap S_2|}.$$
In this case, from (\ref{twosimcom0}) and (\ref{eq:dhinters}) we have
$$d_r(\C_{\Delta^c})=q^m-q^{|S_1|}-q^{|S_2|}-q^{m-r}+q^{|S_1|-r}+q^{|S_2|-r}.$$

From equation (\ref{eq:drc1}) and the discussion above, we obtain Table \ref{table:com2maxied}.
\end{proof}

\begin{example}
Let $\C_{\Delta^c}$ be a linear code defined in (\ref{definingset4}), where $\Delta=\langle S_1,S_2 \rangle $ is a simplicial complex of ${\mathbb F_2^6}$ generated by  support sets $S_1=\{1,2\}$ and $S_2=\{2,3,4\}$. By the help of Magma, we have that $\C_{\Delta^c}$ is a $[54,6,26]$ code and has the weight hierarchy  as follows: $d_1=26, d_2=40, d_3=47,d_4=51, d_5=53, d_6=54$. This result is consistent with Table \ref{table:com2maxied} in Theorem \ref{com2}.
\end{example}


It is known that the weight hierarchy of $\C_{\Delta^c}$ depends on the computation of the minimum value of $|\Delta \cap H^{\perp}|$ for any $H\in [\bF_q^m, r]$. This seems more complicated than computing the maximum value of $|\Delta \cap H^{\perp}|$ for any $H\in [\bF_q^m, r]$ for a general simplicial complex. However, if the generating support sets of a simplicial complex $\Delta$ do not intersect with each other, we can still determine the weight hierarchy of the linear code $\C_{\Delta^c}$ using the following theorem.

\begin{theorem}{\label{ncom}}
Let $\Delta=\langle S_1, S_2, \cdots, S_\ell \rangle$ be a simplicial complex of ${\mathbb F_q^m}$, where $S_1, S_2, \cdots, S_\ell$ are  the support sets of $\ell$ maximal elements of $\Delta$ with $S_i \cap S_j=\emptyset$ for any $1\leq i<j\leq \ell$,  and $1\leq |S_1|\leq  |S_2|\leq \cdots \leq |S_\ell|<m$. Then $\C_{\Delta}$ is a $[q^m-\sum_{i=1}^\ell q^{|S_i|}+\ell-1,m, (q-1)q^{m-1} -\sum_{i=1}^\ell (q-1) q^{|S_i|-1}]$ code and  has the weight hierarchy in Table \ref{table:comnmaxied}.
\begin{table}[h]
     \centering
    \footnotesize
    \renewcommand{\arraystretch}{1.2}
    \setlength{\tabcolsep}{3pt}
	\caption{The weight hierarchy of $\C_{\Delta}$ in Theorem~\ref{ncom}}
    \label{table:comnmaxied}
	\scalebox{1}{
	\begin{tabular}{ll}
	\toprule
	\textbf{The weight hierarchy} \, & \, \textbf{Rang of $r$}\\
	\midrule
	$q^m-\sum_{i=1}^{\ell}q^{|S_i|}-q^{m-r}+\sum_{i=1}^{\ell}q^{|S_i|-r}$ \, & \, $1\leq r< |S_1|$ \\
	$q^m-\sum_{i=1}^{\ell}q^{|S_i|}-q^{m-r}+\sum_{i=2}^{\ell}q^{|S_i|-r}+1$ \, & \, $|S_1|\leq r< |S_2|$ \\
	$\cdots$ \, & \, $\cdots$\\
	$q^m-\sum_{i=1}^{\ell}q^{|S_i|}-q^{m-r}+\sum_{i=j}^{\ell}q^{|S_i|-r}+j-1$ \, &\, $|S_{j-1}|\leq r<|S_j|$\\
	$\cdots$ \,&\, $\cdots$\\
	$q^m-\sum_{i=1}^{\ell}q^{|S_i|}-q^{m-r}+\ell$ \,&\, $|S_{\ell}|\leq r\leq m$\\
	\bottomrule
\end{tabular}
}
\end{table}
\end{theorem}

\begin{proof}
It is easy to see that $\Delta= \cup_{i=1}^\ell \langle S_i\rangle$, where $\langle S_i\rangle$ is a subspace of $\bF_q^m$ with dimension $|S_i|$ for $1\leq i\leq \ell$. The dimension of $\C_{\Delta^c}$ is $m$, which has been shown in \cite{Hu2023}, and the length of code $\C_{\Delta^c}$ is
	$$n=|\Delta^c|=q^m-|\Delta|=q^m-(\sum_{i=1}^\ell q^{|S_i|}-\sum_{1\leq i<j\leq \ell }q^{|S_i\cap S_j|}+\cdots +(-1)^{\ell-1}q^{|S_1\cap S_2\cap \cdots \cap S_\ell |})=q^m-\sum_{i=1}^\ell q^{|S_i|}+\ell-1.$$
By (\ref{eq:drc1}), for any $H\in [{\mathbb F_q^m}, r]$, we have

\begin{equation}{\label{tsimcom}}
\begin{split}
d_r(\C_{\Delta^c})&=|\Delta^c|-\mathop{\rm max}\limits_{H\in [\bF_q^m, r]} (q^{m-r}- |\Delta \cap H^{\perp}|)\\
&=q^m-\sum_{i=1}^\ell q^{|S_i|}+\ell-1- q^{m-r}+\mathop{\rm min}\limits_{H\in [\bF_q^m, r]}\left( |\Delta \cap H^{\perp}|\right),
	\end{split}
	\end{equation}
where $H^{\perp}$ denote the dual space of $H$ and
\begin{equation*}
	\begin{split}
		|\Delta \cap H^{\perp}| &=\sum_{i=1}^\ell|\langle S_i \rangle \cap H^{\perp}|-\sum_{1\leq i<j\leq \ell}|\langle S_i \rangle \cap \langle S_j \rangle \cap H^{\perp}|+\cdots +(-1)^{\ell-1}|\langle S_1 \rangle \cap \cdots \cap \langle S_\ell \rangle \cap H^{\perp}| \\
		&=\sum_{i=1}^\ell |\langle S_i \rangle \cap H^{\perp}|-\ell +1.
	\end{split}
\end{equation*}
From equation (\ref{tsimcom}), we will discuss the minimum value of $|\Delta \cap H^{\perp}|$ for all $r$-dimensional subspaces of ${\mathbb F_q^m}$.

It is known that $\Delta= \cup_{i=1}^\ell \langle S_i\rangle$, which is contained in $S=\sum_{i=1}^\ell \langle S_i\rangle$. Since $|S_1|\leq |S_2|\leq \cdots \leq |S_\ell |$, by Remark~\ref{rem:V1Vn} we can find a largest subspace $W\subseteq \{ {\bf 0 } \} \cup S \setminus (\cup_{i=1}^\ell \langle S_i\rangle)$, which has the dimension $\sum_{i=1}^{\ell-1} |S_i|$, and the dual $S^{\perp}$ of $S$ has the dimension $m-\sum_{i=1}^\ell |S_i|$ since $S_i \cap S_j=\emptyset$ for any $1\leq i<j\leq \ell$. It is clear that $ W \cap S^\perp = \emptyset$ and
$$ {\rm dim}_{\bF_q}\,(W+S^{\perp})=m-|S_\ell|, \,\,  (W+S^{\perp})\cap \langle S_i\rangle =\{ {\bf 0} \} \ {\rm for} \ i=1, 2, \cdots \ell. $$
So, we have
$$\langle S_\ell \rangle^{\perp} = W+S^{\perp}\ {\rm and }\ \langle S_\ell \rangle^{\perp} \subseteq \langle S_{\ell -1}\rangle^{\perp} \subseteq \cdots \subseteq \langle S_1\rangle^{\perp}.$$

If $0\leq m-r\leq m-|S_\ell|$, i.e., $|S_\ell| \leq r\leq m$, then we can find an $H \in [{\mathbb F_q^m}, r]$ such that $H^{\perp} \subseteq W+S^{\perp}$ and $\langle S_i\rangle \cap H^{\perp}=\{ {\bf 0} \}$ for $i=1, 2, \cdots,\ell $. This $H$ minimizes $|\Delta \cap H^{\perp}|$ and ${\rm min}_{H\in [\bF_q^m, r]}\left( |\Delta \cap H^{\perp}|\right)=1$. From (\ref{tsimcom}) we have
$$d_r(\C_{\Delta^c})=q^m-\sum_{i=1}^\ell q^{|S_i|}-q^{m-r}+\ell. $$

If $m-|S_\ell| < m-r\leq m-|S_{\ell-1}|$, i.e., $|S_{\ell-1}|\leq r<|S_\ell|$, then $H^{\perp} \cap \langle S_\ell \rangle \neq \{ {\bf 0} \}$ since
$${\rm dim}_{\bF_q}\, H^{\perp}+{\rm dim}_{\bF_q}\, \langle S_\ell \rangle = m-r+|S_\ell| >m. $$
In this case, we can find an $H\in [\bF_q^m, r]$ such that
$$H^{\perp}=\langle S_\ell \rangle^{\perp}+H^{\prime} \subseteq \langle S_{\ell-1}\rangle^{\perp},$$
where $H^{\prime}\subseteq \langle S_\ell\rangle$ and has the dimension $m-r-m+|S_\ell|=|S_\ell|-r$, and this $H$ minimizes the intersections as follows:
$|\langle S_i\rangle \cap H^{\perp}|=1$ for $i=1, 2, \cdots, \ell-1$ and $|\langle S_\ell \rangle \cap H^{\perp}|=q^{|S_\ell|-r}$. So,
$$  \mathop{\rm min}\limits_{H\in [\bF_q^m, r]}\left( |\Delta \cap H^{\perp}|\right)=\ell-1+q^{|S_\ell|-r}-\ell+1=q^{|S_\ell|-r}.$$
By (\ref{tsimcom}) we have
$$d_r(\C_{\Delta^c})=q^m-\sum_{i=1}^\ell q^{|S_i|}-q^{m-r}+q^{|S_\ell|-r}+\ell-1.$$
	
Repeating the process above we can determine the minimum value of $|\Delta \cap H^{\perp}|$ for all $r$-dimensional subspaces $H$ with $1\leq r \leq m$. From (\ref{tsimcom}) we  obtain  Table \ref{table:comnmaxied}.
\end{proof}


\begin{example}
Let $\C_{\Delta^c}$ be a linear code defined in (\ref{definingset4}), where $\Delta=\langle S_1, S_2, S_3, S_4 \rangle $ is a simplicial complex of ${\mathbb F_3^5}$ generated by the support sets $S_1=\{1\}, S_2=\{2\}, S_3=\{3\}$ and $S_4=\{4, 5\}$. By the help of Magma, we have that $\C_{\Delta^c}$ is a $[228, 5, 150]$ code and has the weight hierarchy as follows: $d_1=150, d_2=202, d_3=220,d_4=226, d_5=228$. This result is consistent with Table \ref{table:comnmaxied} in Theorem \ref{ncom}.
\end{example}

\section{Concluding remarks}
In this paper, we discussed the generalized Hamming weight of linear codes from defining set construction and determined the weight hierarchy of linear codes constructed from simplicial complexes or their complements in $\bF_q^m$. These simplicial complexes are generalized by support sets of one, two, three, or a finite number of maximal elements. For the general case, it appears to be quite complex to express all the generalized Hamming weights of a linear code constructed from a simplicial complex in relation to the generators of the simplicial complex (i.e., the support sets of maximal elements). Interested readers are cordially invited to investigate the weight hierarchy of linear codes constructed from a general simplicial complex, or its complement in $\bF_q^m$.

\section{Appendix}

In Section 3, we determine the weight hierarchy of linear code $\C_{\Delta}$, where $\Delta$ is a simplicial complex generated by the support sets of three maximal elements, i.e.,
$\Delta=\langle S_1, S_2, S_3 \rangle$ for $|S_1|<|S_2|<|S_3|$. If two or three of the support sets have the same size, then $\C_{\Delta}$ has the different weight hierarchy, and
these cases are presented as follows:

{\bf (1)} If $|S_1|=|S_2|<|S_3|$, then $\C_{\Delta}$ has the weight hierarchy as that in Theorem~\ref{thm3}.

{\bf (2)} If $|S_1|<|S_2|=|S_3|$, then the following statements hold.
\begin{description}
\item [(i)] If $|S_1 \cap S_2|\leq |S_1\cap S_3|\leq |S_2\cap S_3|$, then $\C_{\Delta}$ has the weight hierarchy given in Table \ref{table:3maxied} of Theorem~\ref{thm3}.
\item [(ii)] If $|S_1\cap S_3|\leq |S_1\cap S_2|\leq |S_2\cap S_3|$, then $\C_{\Delta}$ has the weight hierarchy given in Table~\ref{APP1}.
 \item [(iii)] If $|S_1\cap S_2|\leq |S_2\cap S_3|\leq |S_1\cap S_3|$, then $\C_{\Delta}$ has the weight hierarchy given in Table \ref{APP2} of Theorem~\ref{thm3}.
 \item [(iv)] If $|S_2\cap S_3|\leq |S_1\cap S_2|\leq |S_1\cap S_3|$, then $\C_{\Delta}$ has the weight hierarchy given in Table \ref{APP2} of Theorem~\ref{thm3}.
 \item [(v)] If $|S_1\cap S_3|\leq |S_2\cap S_3|\leq |S_1\cap S_2|$, then $\C_{\Delta}$ has the weight hierarchy given in Table \ref{APP3}.
 \item [(vi)] If $|S_2\cap S_3|\leq |S_1\cap S_3|\leq |S_1\cap S_2|$, then $\C_{\Delta}$ has the weight hierarchy given in Table~\ref{APP3}.
\end{description}
 \begin{table}[h!]
     \centering
    \footnotesize
    \renewcommand{\arraystretch}{1.2}
    \setlength{\tabcolsep}{3pt}
	\caption{The weight hierarchy of $\C_{\Delta}$}
    \label{APP1}
	\scalebox{1}{
	\begin{tabular}{ll}
	\toprule
	\textbf{The weight hierarchy} & \textbf{Rang of $r$}\\
	\midrule
    \vspace{0.2cm}
	$q^{|S_1|}-q^{m-r-|S_2|-|S_3|+|S_1\cap S_2|+|S_1\cap S_3|+|S_2\cap S_3|-|\cap_{i=1}^3 S_i|}$ & $1\leq r\leq m-|S_2|-|S_3|+|S_2\cap S_3|$ \\
   \vspace{0.2cm}
	$q^{|S_1|}+q^{|S_3|}-q^{m-r-|S_3|+|S_2\cap S_3|}-q^{|S_1\cap S_2|+|S_1\cap S_3|-|\cap_{i=1}^3 S_i|}$ &  \makecell[l]{$m-|S_2|-|S_3|+|S_2\cap S_3|<r\leq $
	\\$m-|S_3|-|S_1\cap S_3|+|\cap_{i=1}^3 S_i|$}\\
     \vspace{0.2cm}
	\makecell[l]{$q^{|S_1|}+q^{|S_3|}-q^{|S_1\cap S_3|}-q^{m-r-|S_3|+|S_1\cap S_2|}
	-q^{m-r-|S_3|+|S_2\cap S_3|} $ \\
	$+q^{m-r-|S_3|+|\cap_{i=1}^3 S_i|}$} & \makecell[l]{ $m-|S_3|-|S_1\cap S_3|+|\cap_{i=1}^3 S_i| $\\
	 $<r\leq m-|S_3|$} \\
	\vspace{0.2cm}
	\makecell[l]{$q^{|S_1|}+q^{|S_2|}+q^{|S_3|}-q^{|S_1\cap S_2|}-q^{|S_1\cap S_3|}-q^{|S_2\cap S_3|}+q^{|\cap_{i=1}^3 S_i|}$\\
	 $-q^{m-r}$} & $m-|S_3|<r \leq m$\\
	\bottomrule
\end{tabular}
}
\end{table}
\begin{table}[h!]
     \centering
    \footnotesize
    \renewcommand{\arraystretch}{1.2}
    \setlength{\tabcolsep}{3pt}
	\caption{The weight hierarchy of $\C_{\Delta}$}
    \label{APP3}
	\scalebox{1}{
	\begin{tabular}{ll}
	\toprule
	\textbf{The weight hierarchy} & \textbf{Rang of $r$}\\
	\midrule
    \vspace{0.2cm}
	$q^{|S_1|}-q^{m-r-|S_2|-|S_3|+|S_1\cap S_2|+|S_1\cap S_3|+|S_2\cap S_3|-|\cap_{i=1}^3 S_i|}$ & $1\leq r\leq m-|S_2|-|S_3|+|S_2\cap S_3|$ \\
	\vspace{0.2cm}
	\makecell[l]{$q^{|S_1|}+q^{|S_2|}-q^{m-r-|S_3|-|S_1\cap S_3|+|\cap_{i=1}^3 S_i|}$ \\
	$-q^{|S_1\cap S_2|+|S_1\cap S_3|-|\cap_{i=1}^3 S_i|}$} &  \makecell[l]{$m-|S_2|-|S_3|+|S_2\cap S_3|<r< $
	\\$m-|S_1|-|S_3|+|S_2\cap S_3|$}\\
	\vspace{0.2cm}
	\makecell[l]{$q^{|S_2|}-q^{m-r-|S_1|-|S_3|+|S_1\cap S_2|+|S_1\cap S_3|+|S_2\cap S_3|-|\cap_{i=1}^3 S_i|}$ } & \makecell[l]{ $m-|S_1|-|S_3|+|S_2\cap S_3| $\\
	 $\leq r< m-|S_1|-|S_3|+|S_1\cap S_2|$} \\
	\vspace{0.2cm}
	$q^{|S_1|}+q^{|S_2|}-q^{|S_1\cap S_3|+|S_2\cap S_3|-|\cap_{i=1}^3 S_i|}-q^{m-r-|S_3|+|S_1\cap S_2|}$ & \makecell[l]{$m-|S_1|-|S_3|+|S_1\cap S_2|\leq r\leq $ \\
	$ m-|S_3|-|S_1\cap S_3|+|\cap_{i=1}^3 S_i|$ }   \\
    \vspace{0.2cm}
	\makecell[l]{$q^{|S_1|}+q^{|S_2|}-q^{|S_1\cap S_3|}-q^{m-r-|S_3|+|S_1\cap S_2|}-q^{m-r-|S_3|+|S_2\cap S_3|}$ \\
	$+q^{m-r-|S_3|+|\cap_{i=1}^3 S_i|}$ } & \makecell[l]{$m-|S_3|-|S_1\cap S_3|+|\cap_{i=1}^3 S_i| $ \\
	$<r\leq m-|S_3|$ } \\
	\vspace{0.2cm}
	\makecell[l]{$q^{|S_1|}+q^{|S_2|}+q^{|S_3|}-q^{|S_1\cap S_2|}-q^{|S_1\cap S_3|}-q^{|S_2\cap S_3|}+q^{|\cap_{i=1}^3 S_i|}$\\
	 $-q^{m-r}$} & $m-|S_3|<r \leq m$\\
	\bottomrule
\end{tabular}
}
\end{table}

{\bf (3)} If $|S_1|=|S_2|=|S_3|$, then $\C_{\Delta}$ has the weight hierarchy as that in {\bf (2)} except for the cases {\bf (iv)} and {\bf (vi)} of {\bf (2)}, and in both cases, $\C_{\Delta}$ has the weight hierarchy as follows.
 \begin{description}
 	\item [(i)]\ If $|S_2\cap S_3|\leq |S_1\cap S_2|\leq |S_1\cap S_3|$, then $\C_{\Delta}$ has the weight hierarchy given in Table \ref{APP4}.
 	\item [(ii)] If $|S_2\cap S_3|\leq |S_1\cap S_3|\leq |S_1\cap S_2|$, then $\C_{\Delta}$ has the weight hierarchy given in Table \ref{APP5}.
 \end{description}

 \begin{table}[h!]
     \centering
    \footnotesize
    \renewcommand{\arraystretch}{1.2}
    \setlength{\tabcolsep}{3pt}
	\caption{The weight hierarchy of $\C_{\Delta}$}
    \label{APP4}
	\scalebox{1}{
	\begin{tabular}{ll}
	\toprule
	\textbf{The weight hierarchy} & \textbf{Rang of $r$}\\
	\midrule
	\vspace{0.2cm}
	\makecell[l]{$q^{|S_2|}-q^{m-r-|S_1|-|S_3|+|S_1\cap S_2|+|S_1\cap S_3|+|S_2\cap S_3|-|\cap_{i=1}^3 S_i|}$ } & \makecell[l]{ $1\leq r\leq  m-|S_1|-|S_3|+|S_1\cap S_3|$} \\
	\vspace{0.2cm}
	$q^{|S_2|}+q^{|S_3|}-q^{|S_1\cap S_2|+|S_2\cap S_3|-|\cap_{i=1}^3 S_i|}-q^{m-r-|S_1|+|S_1\cap S_3|}$ & \makecell[l]{$m-|S_1|-|S_3|+|S_1\cap S_3|< r\leq $ \\
	$ m-|S_1|-|S_2\cap S_3|+|\cap_{i=1}^3 S_i|$ }   \\
    \vspace{0.2cm}
	\makecell[l]{$q^{|S_2|}+q^{|S_3|}-q^{|S_2\cap S_3|}-q^{m-r-|S_1|+|S_1\cap S_2|}-q^{m-r-|S_1|+|S_1\cap S_3|}$ \\
	$+q^{m-r-|S_1|+|\cap_{i=1}^3 S_i|}$ } & \makecell[l]{$m-|S_1|-|S_2\cap S_3|+|\cap_{i=1}^3 S_i| $ \\
	$<r\leq m-|S_1|$ } \\
	\vspace{0.2cm}
	\makecell[l]{$q^{|S_1|}+q^{|S_2|}+q^{|S_3|}-q^{|S_1\cap S_2|}-q^{|S_1\cap S_3|}-q^{|S_2\cap S_3|}+q^{|\cap_{i=1}^3 S_i|}$\\
	 $-q^{m-r}$} & $m-|S_1|<r \leq m$\\

	\bottomrule
\end{tabular}
}
\end{table}
 	\begin{table}[h!]
     \centering
    \footnotesize
    \renewcommand{\arraystretch}{1.2}
    \setlength{\tabcolsep}{3pt}
	\caption{The weight hierarchy of $\C_{\Delta}$}
    \label{APP5}
	\scalebox{1}{
	\begin{tabular}{ll}
	\toprule
	\textbf{The weight hierarchy} & \textbf{Rang of $r$}\\
	\midrule
	\vspace{0.2cm}
	\makecell[l]{$q^{|S_3|}-q^{m-r-|S_1|-|S_2|+|S_1\cap S_2|+|S_1\cap S_3|+|S_2\cap S_3|-|\cap_{i=1}^3 S_i|}$ } & \makecell[l]{ $1\leq r\leq  m-|S_1|-|S_2|+|S_1\cap S_2|$} \\
	\vspace{0.2cm}
	$q^{|S_2|}+q^{|S_3|}-q^{|S_1\cap S_3|+|S_2\cap S_3|-|\cap_{i=1}^3 S_i|}-q^{m-r-|S_1|+|S_1\cap S_2|}$ & \makecell[l]{$m-|S_1|-|S_2|+|S_1\cap S_2|< r\leq $ \\
	$ m-|S_1|-|S_2\cap S_3|+|\cap_{i=1}^3 S_i|$ }   \\
    \vspace{0.2cm}
	\makecell[l]{$q^{|S_2|}+q^{|S_3|}-q^{|S_2\cap S_3|}-q^{m-r-|S_1|+|S_1\cap S_2|}-q^{m-r-|S_1|+|S_1\cap S_3|}$ \\
	$+q^{m-r-|S_1|+|\cap_{i=1}^3 S_i|}$ } & \makecell[l]{$m-|S_1|-|S_2\cap S_3|+|\cap_{i=1}^3 S_i| $ \\
	$<r\leq m-|S_1|$ } \\
	\vspace{0.2cm}
	\makecell[l]{$q^{|S_1|}+q^{|S_2|}+q^{|S_3|}-q^{|S_1\cap S_2|}-q^{|S_1\cap S_3|}-q^{|S_2\cap S_3|}+q^{|\cap_{i=1}^3 S_i|}$\\
	 $-q^{m-r}$} & $m-|S_1|<r \leq m$\\

	\bottomrule
\end{tabular}
}
\end{table}

\begin{example}
Let $\C_{\Delta}$ be a binary linear code defined in (\ref{definingset3}), where $\Delta=\langle S_1,S_2,S_3\rangle $ is a simplicial complex of ${\mathbb F_2^6}$ generated by the support sets $S_1=\{1,2\}$, $S_2=\{2,3,4,5\}$, $S_3=\{1,3,4,6\}$. By the help of Magma, we know that $\C_{\Delta}$ is a $[29,6,8]$ code and has the weight hierarchy as follows: $d_1=8, d_2=13, d_3=21,d_4=25, d_5=27, d_6=28$. This result is consistent with Table \ref{APP1} in Appendix.
\end{example}

\begin{example}
Let $\C_{\Delta}$ be a binary linear code defined in (\ref{definingset3}), where $\Delta=\langle S_1,S_2,S_3\rangle $ is a simplicial complex of ${\mathbb F_2^7}$ generated by the support sets $S_1=\{1,2,3\}$, $S_2=\{1,2,4,5\}$, $S_3=\{3,4,6,7\}$. By the help of Magma, we know that $\C_{\Delta}$ is a $[33,7,8]$ code and has the weight hierarchy as follows: $d_1=8, d_2=12, d_3=17,d_4=25, d_5=29, d_6=31, d_7=32$. This result is consistent with Table \ref{APP3} in Appendix.
\end{example}

\begin{example}
Let $\C_{\Delta}$ be a binary linear code defined in (\ref{definingset3}), where $\Delta=\langle S_1,S_2,S_3 \rangle $ is a simplicial complex of ${\mathbb F_2^5}$ generated by the support sets $S_1=\{1,2,3\}$, $S_2=\{3,4,5\}$, $S_3=\{1,2,4\}$. By the help of Magma, we know that $\C_{\Delta}$ is a $[17,5,4]$ code and has the weight hierarchy as follows: $d_1=4, d_2=9, d_3=13,d_4=15, d_5=16$. This result is consistent with Table \ref{APP4} in Appendix.
\end{example}

\begin{example}
Let $\C_{\Delta}$ be a binary linear code defined in (\ref{definingset3}), where $\Delta=\langle S_1,S_2,S_3\rangle $ is a simplicial complex of ${\mathbb F_2^6}$ generated by the support sets $S_1=\{1,2,3,5\}$, $S_2=\{1,2,4,5\}$, $S_3=\{3,4,5,6\}$. By the help of Magma, we know that $\C_{\Delta}$ is a $[34,6,8]$ code and has the weight hierarchy as follows: $d_1=8, d_2=18, d_3=26,d_4=30, d_5=32, d_6=33$. This result is consistent with Table \ref{APP5} in Appendix.
\end{example}

\end{document}